\documentclass[12pt]{article}

\usepackage{amsmath}
\usepackage{amsfonts}
\usepackage{amsthm}
\usepackage{amssymb}
\usepackage{url}
\usepackage{hyperref}
\usepackage{color}
\usepackage{tkz-graph}

\input{Qcircuit}

\newcommand{\proj}[2]{\ket{#1}\!\bra{#2}}

\newcommand{\braket}[2]{\left \langle #1 | #2 \right \rangle}

\newcommand{\norm}[1]{\left|\left|#1\right|\right|}
\newcommand{\cclass}[1]{\mathsf{#1}}

\newcommand{\ctl}[1]{\text{CTRL-}#1}

\usepackage{color}

\newtheorem{definition}{Definition}

\newtheorem{theorem}{Theorem}

\newtheorem{lemma}{Lemma}
\newtheorem{corollary}{Corollary}
\newtheorem{procedure}{Procedure}

\title{Interactive proofs for $\cclass{BQP}$ via self-tested graph states}
\author{
Matthew McKague \\ 
Jack Dodd Centre for Quantum Technology \& \\ 
Department of Computer Science \\
University of Otago \\\
\url{matthew.mckague@otago.ac.nz}}

\begin{document}
\maketitle
\begin{abstract}
Using the measurement-based quantum computation model, we construct interactive proofs with non-communicating quantum provers and a classical verifier.  Our construction gives interactive proofs for all languages in $\cclass{BQP}$ with a polynomial number of quantum provers, each of which, in the honest case, performs only a \emph{single measurement.}

Our techniques use self-tested graph states.  In this regard we introduce two important improvements over previous work.  Specifically, we derive new error bounds which scale polynomially with the size of the graph compared with exponential dependence on the size of the graph in previous work.  We also extend the self-testing error bounds on measurements to a very general set which includes the adaptive measurements used for measurement-based quantum computation as a special case.
\end{abstract}

\section{Introduction}
We seek to find interactive proofs between quantum provers and classical verifiers, both limited to polynomial-time calculations.  That is to say, we would like to have a procedure where a classical computer (the ``verifier''), limited to a polynomial number of operations, can query a quantum computer (the ``prover''), also limited to a polynomial number of operations, and tap into its resources in order to perform some computation.  Additionally, if the verifier unhappily interacts with a malicious quantum computer it should be able to detect this and abort the calculation, even if the prover has unlimited computational resources.  To make the challenge less trivial, there should exist interactive proofs for problems that are harder than the verifier could solve by itself and ideally there should exist interactive proofs for any problem that the prover can solve by itself.

This problem is interesting for a variety of reasons.  First, as a complexity theoretic question it has obvious value in further developing the theory of how powerful quantum computers are.  From a practical computing point of view, it would be nice to know whether it would be possible to have cheap classical computers interact with large (and presumably more expensive) quantum ``servers,'' paying for services as required.  Of course the users would like to know that they get their money's worth, and interactive computations can confirm this.  As well, from an experimental point of view, interactive proofs can be used to verify the operation of some experimental apparatus.  This is of particular importance for quantum experiments since it may well be that, for large experiments, it is impossible (in practical terms) to classically compute what the predictions of the quantum model are, leading to questions about the falsifiability of the quantum formalism  \cite{Aharonov:2012:Is-Quantum-Mech}. 

Very little is known about this problem as stated.  Clearly the set of languages recognizable by a poly-time classical verifier and poly-time quantum prover lies somewhere between $\cclass{P}$ and $\cclass{BQP}$ since on one hand the verifier can ignore the prover, and on the other hand the verifier and honest prover together form a poly-time quantum machine.  As well, there do exist interactive proofs for all of $\cclass{BQP}$ since $\cclass{BQP} \subseteq \cclass{PSPACE}$ and $\cclass{PSPACE} = \cclass{IP}$ \cite{Shamir:1992:IP--PSPACE,Lund:1990:Algebraic-metho}, but the known constructions require the prover to solve $\cclass{PSPACE}$-complete problems.  Constructions for particular problems are known (\cite{McKague:2010:Interactive-pro} for example) and of course anything in $\cclass{NP}$ has a trivial interactive proof, a general construction has not yet been found.  

Current techniques 
\cite{
	Acin:2007:Device-Independ,
	Bardyn:2009:Device-independ,
	Magniez:2006:Self-testing-of,
	Mayers:2004:Self-testing-qu,
	Mayers:1998:Quantum-Cryptog,
	McKague:2012:Robust-self-tes,
	McKague:2010:Self-testing-gr,
	Miller:2012:Optimalrobustquantum,
	Pironio:2009:Device-independ,
	Pironio:2010:Random-numbers-,
	Reichardt:2012:A-classical-lea}
for probing the behaviour of adversarial quantum systems all rely on entanglement and hence in order to make use of them we must introduce more provers.  Reichardt et al.\ \cite{Reichardt:2012:A-classical-lea} considered the case of two provers.  Here we will consider the case of a polynomial number of provers, but each limited to a single operation, and show that we can recognize all of $\cclass{BQP}$ with this model.

Our construction uses two major components.  One is self-testing and the other is measurement-based quantum computation.  Self-testing allows us to confirm that the provers hold on to a graph state and perform certain measurements on this state when instructed to do so.  Measurement-based quantum computation allows us to use these verified resources to perform the desired calculation.

\subsection{Previous work}

Self-testing was introduced by Mayers and Yao \cite{Mayers:1998:Quantum-Cryptog, Mayers:2004:Self-testing-qu}.  Their goal was to establish that a pair of devices share a maximally entangled pair of qubits, and that the devices implement some specific measurements, all while making a minimum of assumptions on the devices.  Most importantly they make no assumptions about the dimension of the Hilbert space associated with the devices.  Meanwhile, van Dam et al.\ \cite{van-Dam:1999:Self-Testing-of} considered testing gates in the context of known Hilbert space dimension.  Magniez et al.\ \cite{Magniez:2006:Self-testing-of} combined the two approaches, allowing testing of entire quantum circuits.  Further refinements, including simpler proof techniques and extension to complex measurements appear in \cite{McKague:2010:Quantum-Informa} and \cite{McKague:2011:Generalized-Sel}.  Self-testing of graph states, critical for our application, appears in \cite{McKague:2010:Self-testing-gr}.  Miller and Shi \cite{Miller:2012:Optimalrobustquantum} also give a general construction for self-testing states based on any XOR game.

These previous works all require additional assumptions.  In particular, they assume that devices can be used repeatedly in an independent and identical manner in order to gather necessary statistics.  As well, \cite{Magniez:2006:Self-testing-of} assumes that certain states are in a product form.  McKague and Magniez (in preparation) remove these assumptions for quantum circuits using techniques similar to those used here.

Stemming from a different heritage, Broadbent et al.\ \cite{Broadbent:2008:Universal-blind} considered a semi-quantum verifier who only prepares single qubit states, and a fully quantum prover.  They give a construction for an interactive proof for any language in $\cclass{BQP}$.  Additionally, they describe (without rigorous proof) an extension using two quantum provers and a classical verifier.  Their construction uses measurement-based quantum computation.  Aharonov et al.\ \cite{Aharonov:2008:Interactive-Pro} also describe a semi-quantum protocol using a constant sized quantum verifier and a polynomial-time quantum prover.

In the context of quantum cryptography, Ac\'{i}n et al.\ \cite{Acin:2007:Device-Independ} introduced \emph{device independent quantum key distribution.}  This model is very similar to that used here.  However, rather than computation the goal is to expand a private shared key.  From a physics perspective, Bardyn et al.\ \cite{Bardyn:2009:Device-independ} and McKague et al.\ \cite{McKague:2012:Robust-self-tes} consider self-testing type entanglement tests from the perspective of Bell inequalities.

Most recently, Reichardt et al.\ \cite{Reichardt:2012:A-classical-lea} proved a very general result allowing two non-communicating quantum provers along with a classical verifier to recognize all of $\cclass{BQP}$.  The core of their result is a self-test, using only two provers, for multiple EPR pairs and measurements. Using this tool they show how to test individual gates and perform measurements via teleportation.  Finally, they combine the results to give an interactive proof for entire quantum circuits.

Measurement-based quantum computation, also known as one-way quantum computation or graph state computation, was introduced by Raussendorf and Briegel \cite{Raussendorf:2001:A-One-Way-Quant, Raussendorf:2003:Measurementbasedquantum}.  In this model of computation we begin with a  graph state and perform measurements on each vertex, with the sequence of vertices and the measurement bases used determined by the calculation we wish to make.  The outcome of the calculation is then derived from the measurement outcomes.  One important aspect of the measurements is that they are adaptive - the measurement basis for a particular vertex can depend on the outcomes of measurements on previous vertices.  This allows us to perform any calculation in $\cclass{BQP}$.  The particular variety of graph-state computation that we use is due to Mhalla and Perdrix \cite{Mhalla:2012:Graph-States-Pi}. The advantage of this model is that it requires measurements in the $X$-$Z$ plane only.

\subsection{Contributions}

We make several important contributions.  First, we modify the proof for the graph state self-test from \cite{McKague:2010:Self-testing-gr}, allowing a tighter error analysis.  For graphs on $n$ vertices the error in the state is upper bounded by $O(\sqrt{n})\epsilon^{\frac{1}{4}}$ (where $\epsilon$ bounds the noise in the experimental outcomes) rather than $O(2^{\frac{n}{2}}) \epsilon^{\frac{1}{2}}$ as in \cite{McKague:2010:Self-testing-gr}.  This exponential improvement in the error scaling in $n$ makes it possible to self-test with a polynomial number of trials to achieve a constant error.  We also analyse the error in the case of adaptive measurements, which are required for measurement-based quantum computing.  Additionally we extend the graph state test to $X$-$Z$ plane measurements  in order to achieve universal computation.  Finally we show how to use the self-test in order to test the provers for honesty in the interactive proof scenario.  Combining this test for honesty with measurement-based quantum computation we achieve the following theorem:

\begin{theorem}
\label{theorem:bqpinteractiveproofs}
For every language $L \in \cclass{BQP}$ and input $x$ there exists a poly$(|x|)$-time verifier $V$ which interacts with a poly$(|x|)$ number of non-communicating quantum provers such that
\begin{itemize}
	\item If $x \in L$ then there exists\footnote{
			The honest provers and the verifier are, of course, members of a uniform set, i.e. a description of the verifier and provers can be generated by a polynomial-time Turing machine.
		} 
		a set of honest quantum provers, each of which performs a single operation, for which $V$ accepts with probability at least $c=2/3$.
	\item If $x \notin L$ then, for any set of provers, $V$ accepts with probability no more than $s=1/3$.
\end{itemize}
\end{theorem}

Along the way we we also prove several results which may be of independent interest.  In particular our error analysis for triangular cluster states can be applied to general graph states and stabilizer states enabling self-testing of these states with robust error bounds.  As well, our error bounds for adaptive measurements are quite general, applying to general quantum circuits which incorporate the untrusted measurements performed by the provers.

Compared with the result of Reichardt et al.\ \cite{Reichardt:2012:A-classical-lea} our contribution is to provide a different construction with different underlying computational model, that of measurement-based quantum computation.  While they use a constant number of provers, each of which runs in polynomial time, we use a polynomial number of provers, each of which runs in constant time (indeed, each prover only performs a single measurement).  The advantage of our technique is that, since only measurements are used, there is no need for any process tomography.  As well, the provers are very easy to implement, requiring only the ability to measure in four different bases (once an appropriate graph state is prepared).  Finally, there is a very nice conceptual advantage, which is that \emph{the measurement-based calculation that is performed is exactly what would be done with trusted devices}, whereas the Reichardt et al.\ construction requires qubits to be teleported between the two provers at each gate.

\subsection{Overview of construction}

We can divide our interactive proof into two distinct units: the calculation and the test for honesty.  The calculation is exactly the same measurement-based quantum computation that would be performed for trusted devices.  The test for honesty is derived from self-testing.

We give some technical details of measurement-based quantum computation in section~\ref{sec:mbqcoverview}.  The procedure can be summarized as:
\begin{procedure}\ 
	\begin{enumerate}
		\item Prepare a universal graph state
		\item Perform measurements to obtain a computation-specific graph state
		\item Measure vertices in sequence, adapting bases according to outcomes from previous measurements
		\item Calculate the final outcome
	\end{enumerate}
\end{procedure}
In order to perform the computation we need the provers to share a graph state and be able to measure vertices.  The verifier performs all the classical computation, including deriving the measurement patterns, the required graph state, and the final outcome.

Our main contributions lie in constructing a test for honesty.  Here we must define some test such that if the provers were to cheat on the calculation then they will fail the test.  Our test for honesty is based on the graph state self-test, originally presented in \cite{McKague:2010:Self-testing-gr}.  It allows the verifier to establish that the provers have access to high quality copies of the desired graph state and $X$ and $Z$ Pauli measurements.  We give details for this test, including our improved proof in section~\ref{sec:graphstatetest}.  

In addition, for the measurement-based quantum computation we also need measurements covering the entire $X$-$Z$ plane.  This is a simple extension of the graph-state test, which we present in section~\ref{sec:xyplanemeasurements}.

The graph-state test, with extensions, define a set of subtests, each of which the provers must pass.  To administer the entire test, the verifier just chooses one of these subtests at random.  If the provers actually hold the required graph state and perform the measurements faithfully then they will pass the test with high probability, and if their behaviour deviates too much from the honest provers then they will pass with a lower probability.  The gap is $1/poly(n)$ for a constant error bound and is calculated in section~\ref{sec:oneshottest}.

With all of this in place we obtain a simple statement:  if the provers deviate from the honest behaviour by more than $\delta$ (see section~\ref{sec:defclose} for a definition), then they will pass the test with probability at most $c_{test} - \epsilon$, where $\epsilon$ is a function of $\delta$ and $c_{test}$ is the probability of honest provers passing the test.  Hence if the provers attempt to cheat we will catch them.  The details are given in section~\ref{sec:oneshottest}.

Having shown how to test whether the provers are honest, and how to perform the desired calculation, we must put these two components together to form the interactive proof.  The structure is as follows: randomly either check for honesty or perform the calculation.  The critical observation is that the queries to  an individual prover look the same whether the verifier is testing or calculating.  More specifically, every query that appears as part of a calculation also appears as part of the test for honesty.  Hence \emph{provers who attempt to cheat on the calculation can be caught by the test for honesty.}

The final technical piece of the puzzle is to determine with what probability to test for honesty.  We give the derivation in section~\ref{sec:interactiveproof}.

\section{Technical introduction}
\label{sec:technicalintro}
In this section we present some notation and definitions used in the construction and proof.  Further technical results are collected in appendix~\ref{appendix:technicallemmas} for convenience.

\subsection{Measurement-based Quantum Computation}
\label{sec:mbqcoverview}

Here we give a general overview of measurement-based quantum computation. Our goal is to provide sufficient background for readers to understand the major features of measurement-based quantum computation (MBQC).  For more detail we refer the reader to \cite{Raussendorf:2001:A-One-Way-Quant, Raussendorf:2003:Measurementbasedquantum}.

To understand how MBQC works, we will show how to turn a simple teleportation circuit into a circuit that applies a gate encoded in a measurement angle.  Let us start with a basic teleportation circuit as in figure~\ref{fig:teleportqubit}.  Rather than performing entanglement swapping with an EPR pair held in memory, as in the usual case, we entangle the input and output qubits directly using a $\ctl{X}$ gate.  The classical result of the measurement in the $X$ basis is used to control a $Z$ gate, which applies a necessary correction.  Direct calculation shows that the input state appears in the output register after the circuit is applied.  In the second circuit in figure~\ref{fig:teleportqubit}, we convert the $\ctl{X}$ gate to a $\ctl{Z}$ gate and two Hadamard gates.  In the third circuit in figure~\ref{fig:teleportqubit}, the left Hadamard simply changes the initial state from $\ket{0}$ to $\ket{+}$.  We move the right Hadamard past the $Z$ correction, which then becomes an $X$ correction gate.

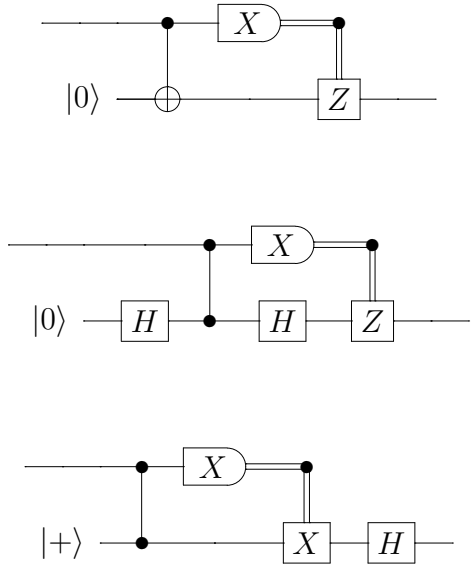
\begin{figure}[h]
\[
	\Qcircuit @C=0.5cm @R=0.5cm{
		& \qw & \qw &  \ctrl{1} & \measureD{X} & \control \cw \cwx[1] \\
		& & \lstick{\ket{0}}  & \targ \qw & \qw & \gate{Z} & \qw & \qw 
	}
\]
\vspace{1cm}
\[
	\Qcircuit @C=0.5cm @R=0.5cm{
		& \qw & \qw & \qw &  \ctrl{1} & \measureD{X} & \control \cw \cwx[1] \\
		&  & \lstick{\ket{0}} & \gate{H}  & \control \qw & \gate{H} & \gate{Z} & \qw & \qw 
	}
\]
\vspace{1cm}
\[
	\Qcircuit @C=0.5cm @R=0.5cm{
		& \qw & \qw &  \ctrl{1} & \measureD{X} & \control \cw \cwx[1] \\
		& & \lstick{\ket{+}}  & \control \qw & \qw & \gate{X} & \gate{H} & \qw 
	}
\]
\caption{Three equivalent basic teleportation circuits.  In the second circuit the $\ctl{X}$ gate is replaced with a $\ctl{Z}$ gate sandwiched between two Hadamard gates.  In the third circuit the left Hadamard gate changes $\ket{0}$ to $\ket{+}$ and the right Hadamard gate moves past the $Z$ correction, changing it to an $X$.}
\label{fig:teleportqubit}
\end{figure}

Now suppose that we apply a unitary $U$ to the qubit as in figure~\ref{fig:teleportqubitunitary}.  For this construction we suppose that $U(\theta) = \exp(\frac{i \theta Z}{2})$ so that it commutes with the $\ctl{Z}$ as in the second circuit of figure~\ref{fig:teleportqubitunitary}.  Now we can see $U$ as a modification of the measurement basis as in the final circuit.  Since we originally measured in the $X$ basis the new measurement basis will be in the $X$-$Y$ plane of the Bloch sphere:  $U^\dagger XU = R(\theta) = \cos \theta\, X + \sin \theta\, Y$.  

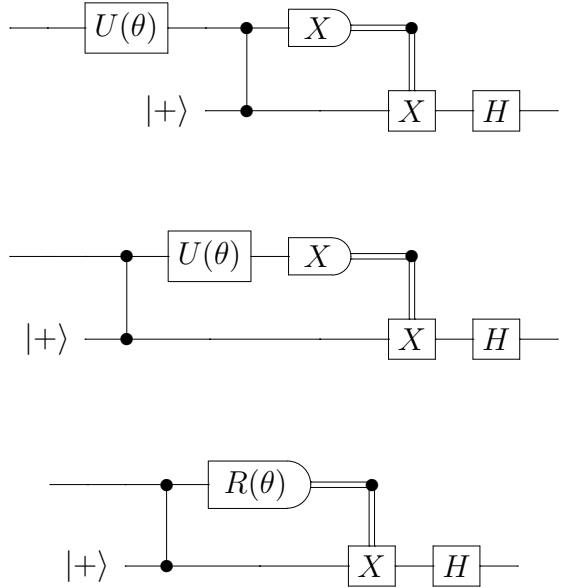
\begin{figure}[h]
\[
	\Qcircuit @C=0.5cm @R=0.5cm{
		& \qw & \gate{U(\theta)} & \qw &  \ctrl{1} & \measureD{X} & \control \cw \cwx[1] \\
		& & & \lstick{\ket{+}}  & \control \qw & \qw & \gate{X} & \gate{H} & \qw 
	}
\]
\vspace{1cm}
\[
	\Qcircuit @C=0.5cm @R=0.5cm{
		& \qw  & \qw &  \ctrl{1} & \gate{U(\theta)} & \measureD{X} & \control \cw \cwx[1] \\
		& & \lstick{\ket{+}}  & \control \qw & \qw & \qw & \gate{X} & \gate{H} & \qw 
	}
\]
\vspace{1cm}
\[
	\Qcircuit @C=0.5cm @R=0.5cm{
		& \qw & \qw &  \ctrl{1} & \measureD{R(\theta)} & \control \cw \cwx[1] \\
		& & \lstick{\ket{+}}  & \control \qw & \qw & \gate{X} & \gate{H} & \qw 
	}
\]
\caption{Three equivalent circuits combining a unitary with teleportation.  In the second circuit the fact that $U(\theta) = \exp{i \frac{\theta Z}{2}}$ is diagonal means that it commutes with the $\ctl{Z}$ gate.  In the third circuit the $U(\theta)$ gate has modified the measurement basis to $R(\theta) = \cos \theta \, X + \sin \theta \, Y$.}
\label{fig:teleportqubitunitary}
\end{figure}

\begin{figure}[h!]
\[
	\Qcircuit @C=0.5cm @R=0.5cm{
		& \qw & \qw &  \ctrl{1} & \measureD{R(\theta_1)} & \control \cw \cwx[1] \\
		& &  \lstick{\ket{+}}  & \control \qw & \qw & \gate{X} & \ctrl{1} & \measureD{R(\theta_2)} & \control \cw \cwx[1] \\
		& & & & & \lstick{\ket{+}}  & \control \qw & \qw & \gate{X} & \qw  \\ 
	}
\]
\vspace{1cm }
\[
	\Qcircuit @C=0.5cm @R=0.5cm{
		& \qw & \qw &  \ctrl{1} & \measureD{R(\theta_1)} & \control \cw \cwx[1] \\
		& &  \lstick{\ket{+}}  & \control \qw & \ctrl{1} & \gate{X} \cwx[1]  & \measureD{R(\theta_2)} & \control \cw \cwx[1] \\
		& &  \lstick{\ket{+}}  & \qw & \control \qw & \gate{Z}	& \qw  &  \gate{X} & \qw  \\ 
	}
\]
\vspace{1cm }
\[
	\Qcircuit @C=0.5cm @R=0.5cm{
		& \qw & \qw &  \ctrl{1} & \measureD{R(\theta_1)} & \control \cw \cwx[2] & \control \cw \cwx[1] \\
		& &  \lstick{\ket{+}}  & \control \qw & \ctrl{1} & \qw  & \measureD{R(\pm \theta_2)} & \control \cw \cwx[1] \\
		& &  \lstick{\ket{+}}  & \qw & \control \qw & \gate{Z}	& \qw  &  \gate{X} & \qw  \\ 
	}
\]
\caption{Two cascaded teleportations.  The first circuit teleports the first qubit to the third, applying $H U(\theta_2) H U(\theta_1)$.  In the second circuit we have moved the $\ctl{Z}$ to the left past the $X$ correction, inducing a $Z$ correction on the third qubit, but allowing all the $\ctl{Z}$ gates to be applied before any measurements are made.  Finally, since $XR(\theta)X = R(-\theta)$ the $X$ correction can be omitted in favour of a change of measurement basis.}
\label{fig:cascadedteleportunitary}
\end{figure}
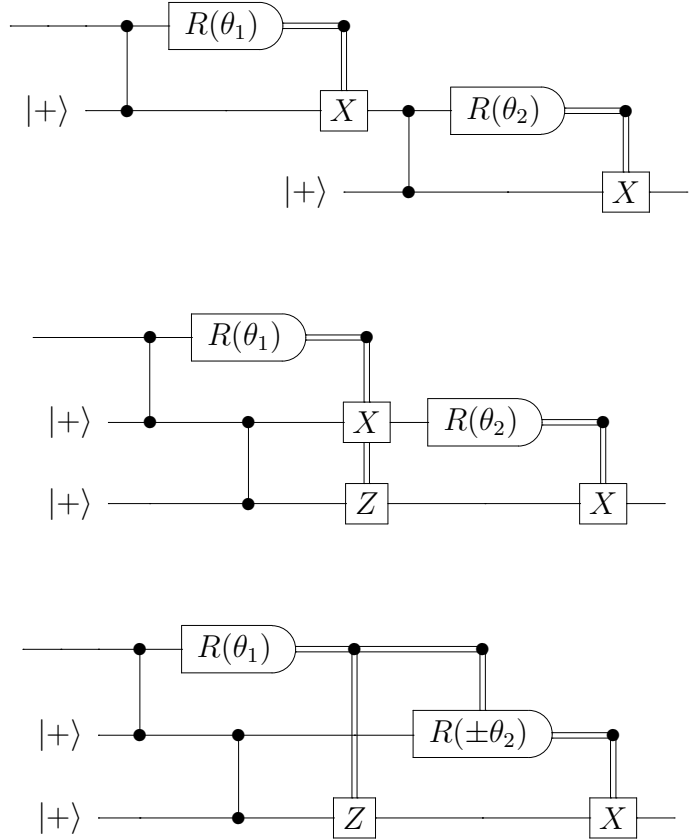

Next we consider how multiple teleportations work together.  First we consider the case of two cascaded teleportations as in figure~\ref{fig:cascadedteleportunitary}.  Using measurement angles $\theta_1$ and $\theta_2$, the overall unitary applied by the circuit is $H U(\theta_2) H U(\theta_1)$.  In the second circuit of figure~\ref{fig:cascadedteleportunitary} we have moved the second $\ctl{Z}$ gate, used to entangled the second and third qubits together, to the left past the $X$ correction on the second qubit.  This induces a $Z$ correction on the third qubit, controlled along with the $X$ correction.  Finally, in the third circuit we incorporate the $X$ correction into the measurement angle on the second qubit.  Indeed, since $XR(\theta)X = R(-\theta)$, the angle $\theta_2$ becomes $-\theta_2$ whenever an $X$ correction is needed.

We have seen how to convert $X$ corrections into changes in the measurement angle.  $Z$ corrections are even easier to apply.  Since $ZR(\theta)Z = - R(\theta)$, a $Z$ correction corresponds to simply inverting the output of a measurement.  Figure~\ref{fig:xzcorrectionmeasurement} shows how $X$ and $Z$ corrections together modify the behaviour of the measurement.

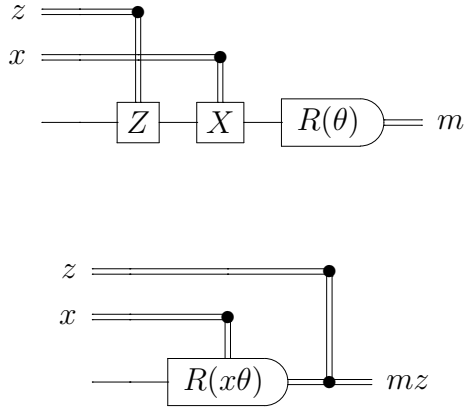
\begin{figure}[h!]
\[
	\Qcircuit @C=0.5cm @R=0.5cm{
		\lstick{z} & \cw & \control \cw \cwx[2]  \\
		\lstick{x} & \cw & \cw & \control \cw \cwx[1] \\
		& \qw & \gate{Z} & \gate{X} & \measureD{R(\theta)} &  \rstick{m} \cw
	}
\]
\vspace{1cm}
\[
	\Qcircuit @C=0.5cm @R=0.5cm{
		\lstick{z} & \cw & \cw & \control \cw \cwx[2]  \\
		\lstick{x} & \cw & \control \cw \cwx[1] \\
		& \qw & \measureD{R(x \theta)} & \control \cw &  \rstick{mz} \cw
	}
\]
\caption{Incorporating $X$ and $Z$ corrections into measurements.  We have $X$ and $Z$ corrections according to some previous measurement results $x,z \in \{\pm 1\}$.  The $X$ correction is incorporated into the measurement as a change in the angle.  The $Z$ correction is incorporated by flipping the outcome of the measurement.}
\label{fig:xzcorrectionmeasurement}
\end{figure}

So far our construction has the following features: we can apply a sequence of unitaries $HU(\theta_n) \dots HU(\theta_1)$ to a qubit by repeatedly teleporting the qubit and varying the measurement angle used in the teleportation.  The necessary corrections from the teleportation can be incorporated into subsequent measurement angles and outcomes, and all the entangling $\ctl{Z}$ gates can be pushed to the start of the procedure.  Hence we can perform a single qubit circuit by first building a large entangled state using $\ket{+}$ states and $\ctl{Z}$ gates, and then measuring the qubits in sequence, adapting measurement angles as we go.  Note that the gates $HU(\theta)$ form a universal set.

In order to perform general circuits we need one more piece of the puzzle, which is two-qubit gates.  In this case we obtain universality by including $\ctl{Z}$ gates.  These can be applied at any time during the circuit and appear as additional $\ctl{Z}$ gates on target qubits when we translate into the teleportation scheme.  These can be treated similarly to the $\ctl{Z}$ gates which are used to entangle input and output qubits for teleportation.  In particular, we can push the $\ctl{Z}$ gates back to the beginning of the circuit, past $X$ and $Z$ corrections.  This induces extra corrections which must be taken into account on subsequent measurements.

Now we have the complete picture.  A calculation begins by preparing many $\ket{+}$ states and entangling them with $\ctl{Z}$ gates.  Then they are measured one at a time, and measurements are adjusted to incorporate $X$ and $Z$ corrections as required.

The initial state, prepared by applying $\ctl{Z}$ gates to qubits in the $\ket{+}$ state, is called a \emph{graph state} and will play an important role in our results here.

Our construction will use a slightly different model of measurement-based quantum computation.  Although the usual and most easily understood method utilises measurements in the $X$-$Y$ plane, we will instead use a different model, due to Mahalla and Perdrix \cite{Mhalla:2012:Graph-States-Pi}, which requires only $X$-$Z$ plane measurements.  In particular they prove the following theorem:

\begin{theorem}[Mahalla and Perdrix \cite{Mhalla:2012:Graph-States-Pi}] \label{theorem:graphstatecomputation}
Triangular cluster states are universal resources for measurement-based computation based on $X$-$Z$ plane measurements.
\end{theorem}

Triangular cluster states are graph states where the underlying graph is a triangular lattice.  As we shall see, these particular graph states are particularly easy to self-test since every vertex is in a triangle.  The proof of the above theorem consists of two parts: showing that triangular cluster states can be converted into other graph states using measurements alone, and showing that $X$-$Z$ measurements suffice for universal computation.  The details of the proof are not important for our results here.  What is important is that the overhead introduced by the construction is small, so that a given quantum circuit gets translated into a graph state with size polynomial in the size of the original circuit.

\subsection{Operators, isometries, bit strings}

We will frequently deal with a tensor product of operators over several subsystems.  To make this easier we use the following notation:  

\begin{definition}
Given some collection of operators $\{M_{j}: j = 1 \dots n\}$ with $M_{j}$ operating on the $j$-th subsystem, and a vector $x \in \{0,1\}^{n}$ define
\begin{equation}
	M^{x} = \bigotimes_{j=1}^{n} M_{j}^{x_{j}}.
\end{equation}
\end{definition}

This notation is quite frequently used with Pauli operators, but here we do not assume that the $M_{j}$ operators are all the same.  Instead, we merely suppose that there is some common label ``$M$'', which may refer to different operators on different subsystems.

Another set of objects that we will deal with frequently is \emph{isometries}.
\begin{definition}
An \emph{isometry} is a linear operator $\Phi: \mathcal{X} \rightarrow \mathcal{Y}$ that preserves inner products.
\end{definition}

Isometries are a natural generalization of unitaries where the image space of $\Phi$ is not necessarily the same as $\mathcal{X}$, and may in general have a larger dimension.  As a concrete and pertinent example, adding an ancilla prepared in a particular state and applying a unitary are both isometries, as is their composition.  Isometries are naturally extended to the dual space by $\Phi(\bra{\psi}) = \Phi(\ket{\psi})^{\dagger}$ and to operators by $\Phi(\proj{x}{y}) = \Phi(\ket{x})\Phi(\bra{y})$, combined with linearity.

As we shall see, we will need to address the state spaces of provers individually, se we will need the concept of a \emph{local} isometry.

\begin{definition}
A \emph{local isometry} on $n$ subsystems is an isometry of the form
\begin{equation}
	\Phi = \Phi_{1} \otimes \dots \otimes \Phi_{n}
\end{equation}
where $\Phi_{j}$ operates on the $j$-th subsystem only.
\end{definition}
Here a tensor product of isometries is evaluated in a way analogous to how a tensor product of unitaries is applied:  decompose the state into a sum of product states and apply the operator to the appropriate vector in the tensor product.  That is to say,
\begin{equation}
	\Phi_{1} \otimes \Phi_{2}
	\left(
		\sum_{j}
			\ket{x_j}_1 
			\ket{y_j}_2
	\right) 
	= 
	\sum_{j}
	\Phi_1
	\left(
		\ket{x_j}_1
	\right) 
	\otimes 
    \Phi_2
    \left(
        \ket{y_j}_2.
    \right)
\end{equation}

By convention, we take $\Phi_1$ to mean $\Phi_1 \otimes I_2$ when applied to a state in $\mathcal{H}_1 \otimes \mathcal{H}_2$, and analogously for other product spaces.

From this it is easy to derive the following properties of local isometries.

\begin{lemma}
\label{lemma:isometryproperites}
Let $\Phi = \Phi_1 \otimes \Phi_2$ be a local isometry, $\ket{\psi}_{1,2}$ be a bipartite state, and $M_1$ be a local operator on the first subsystem.  Then
\begin{eqnarray}
	\Phi(M_1 \ket{\psi_{1,2}}) = \Phi(M_1) \Phi(\ket{\psi_{1,2}})
\\	
	\Phi(M_1 \ket{\psi_{1,2}}) = \Phi_1(M_1) \Phi(\ket{\psi_{1,2}})
\end{eqnarray}

\end{lemma}
We make extensive use of bit strings.  For an $n$-bit string $t$ the $j$-th bit is $t_{j}$.  Inner products of bit strings are given by
\begin{equation}
	s \cdot t = \sum_{j = 1}^{n} s_{j} t_{j}.
\end{equation}
We will, at times, consider the inner product as an integer, and at other times as a bit (i.e. over $\mathbb{Z}$ or $\mathbb{Z}_{2}$).  Where the difference is important we will specify.  For example, $t \cdot t$ taken over $\mathbb{Z}$ gives the number of ones in $t$ but when taken over $\mathbb{Z}_2$ it is the parity of the number of ones. 

Finally, we define the bit string $1_{v}$ to have a $1$ only in the $v$ position and zeros elsewhere, i.e. $(1_v)_j = \delta_{vj}$.  

\subsection{Graph states}
We assume that the reader is familiar with the basics of graph theory.  A good resource is \cite{Diestel:2010:Graph-Theory}.  We now fix some notation for our convenience.  Let $G = (V, E)$ be a graph, $n = |V|$ and $u, v \in V$.  The \emph{adjacency matrix} $\mathbf{A}$ of $G$ is a $\{0,1\}$ matrix with $\mathbf{A}_{u,v} = 1$ whenever $(u,v) \in E$ and 0 elsewhere.  Note that $\mathbf{A}1_{v}$ is a vector containing a 1 in position $u$ for each $(u,v) \in E$, and is hence the characteristic vector of the neighbourhood of $v$.  A \emph{subgraph} of $G$ is a graph with vertices $V^\prime \subseteq V$ and edges $E^\prime \subseteq E$ such that all edges in $E^\prime$ go between vertices of $V^\prime$. Finally, the \emph{induced subgraph} on a subset $S \subseteq V$ is the graph on vertices $S$ which has edges $\{(u,v) | u,v \in S, (u,v) \in E\}$.  In other words, the induced subgraph is the \emph{maximal subgraph} of $G$ on vertices in $S$.  A \emph{triangle} is a set of three vertices which are pairwise adjacent.

The graph state $\ket{G}$ is an $n$-qubit state, with qubits labelled by vertices, which is stabilized\footnote{
See~\cite{Gottesman:1997:Stabilizer-Codes-and-Quantum-Error-Correction} for more information on the stabilizer formalism.
} by the operators
\begin{equation}
	S_{v} = X_{v} Z^{\mathbf{A}1_{v}}
\end{equation}
where
\begin{equation}
    X = 
    \left(
        \begin{matrix}
            0 & 1 \\
            1 & 0                     
    \end{matrix}
    \right)
    ,
    \,\, \, 
        Z = 
    \left(
        \begin{matrix}
            1 & 0 \\
            0 & -1                     
    \end{matrix}
    \right).
\end{equation}
That is, $S_{v}$ has $X$ on vertex $v$ and $Z$ on each of its neighbours and
\begin{equation}
	S_{v} \ket{G} = \ket{G}.
\end{equation}
Equivalently,
\begin{equation}
	\ket{G} = 
	\frac{1}{\sqrt{2^{n}}}
	\sum_{x \in\{0,1\}^{n}} 
		(-1)^{\frac{1}{2}x \cdot \mathbf{A}x}\ket{x}
\end{equation}
with the inner product over $\mathbb{Z}$.  To explain, let us write
\begin{equation}
	x \cdot \mathbf{A}x = \sum_{\substack{u,v \\ x_u = 1 = x_v }} 1_u \cdot \mathbf{A} 1_v = \sum_{\substack{u,v \\ x_u = 1 = x_v}} \mathbf{A}_{u,v}
\end{equation}
Now since $\mathbf{A}_{u,v} = \mathbf{A}_{u,v}=1$ whenever $(u,v) \in E$, we are counting edges.  The summation and $\mathbf{A}$ are symmetric, so we are double counting and we always get an even number (hence the $\frac{1}{2}$ appearing in the exponent above).  Let $T_x = \{v | x_v = 1\}$, then we are summing over all the vertices in $T_x$, double counting the edges in the induced subgraph on $T_x$.
 
For completeness we show that the above two definitions are equivalent by showing that $\ket{G}$ is stabilized by $S_v$:
\begin{eqnarray}
	X_{v} Z^{\mathbf{A}1_{v}} \ket{G} & = &
	\frac{1}{\sqrt{2^{n}}} 
	\sum_{x} 
		(-1)^{\frac{1}{2}x \cdot \mathbf{A}x}
		(-1)^{ x \cdot \mathbf{A}1_{v}} 
		\ket{x \oplus 1_{v}} \\
	& = &
	\frac{1}{\sqrt{2^{n}}} 
	\sum_{x} 
		(-1)^{
			\frac{1}{2}
			(x  \oplus 1_{v})
			\cdot \mathbf{A}
			(x  \oplus 1_{v})
		\pm
			(x  \oplus 1_{v} )
			\cdot \mathbf{A}1_{v}
		} 
		\ket{x}
\end{eqnarray}
where we have re-indexed the summation by $x \rightarrow x \oplus 1_{v}$.  The $\pm$ in the exponent of the $-1$ represents the fact that we only care about the parity of the exponent, so we can add or subtract as we please.

Now $
	\frac{1}{2}
		(x  \oplus 1_{v})
		\cdot \mathbf{A}
		(x  \oplus 1_{v})
$
is the number of edges in the induced subgraph on $T_{x \oplus 1_v}$.  Meanwhile
$
	(x  \oplus 1_{v} )
	\cdot \mathbf{A}1_{v} = x \cdot \mathbf{A}1_v
$
since $1_v \cdot \mathbf{A} 1_v = 0$ (no vertex is adjacent to itself)  and 
$x \cdot \mathbf{A}1_v$ counts the neighbours of $v$ that are in $S_{x}$. 

There are two cases.  First, if $v \in T_x$ then $T_{x \oplus 1_v}$ does not contain $v$.  The subgraph on $T_{x}$ is obtained from the induced subgraph on $T_{x \oplus 1_v}$ by adding $v$ and all the associated edges -  $x \cdot \mathbf{A}1_v$ of them - and the total number of edges in the induced subgraph on $T_x$ is
\begin{equation}
\label{eq:adjmatrixadd}
		\frac{1}{2}
			(x  \oplus 1_{v})
			\cdot \mathbf{A}
			(x  \oplus 1_{v})
		+
			(x  \oplus 1_{v} )
			\cdot \mathbf{A}1_{v}
		=
			\frac{1}{2}
			x \cdot \mathbf{A} x.
\end{equation}
In the other case $v \notin T_x$, so we obtain $T_x$ by removing $v$ and all associated edges from $T_{x \oplus 1_v}$, so
\begin{equation}
\label{eq:adjmatrixadd2}
		\frac{1}{2}
			(x  \oplus 1_{v})
			\cdot \mathbf{A}
			(x  \oplus 1_{v})
		-
			(x  \oplus 1_{v} )
			\cdot \mathbf{A}1_{v}
		=
			\frac{1}{2}
			x \cdot \mathbf{A} x.
\end{equation}
Hence
\begin{eqnarray}
	X_{v} Z^{\mathbf{A}1_{v}} \ket{G} & = &
	\frac{1}{\sqrt{2^{n}}} 
	\sum_{x} 
		(-1)^{
			\frac{1}{2}
			(x  \oplus 1_{v})
			\cdot \mathbf{A}
			(x  \oplus 1_{v})
		\pm
			(x  \oplus 1_{v} )
			\cdot \mathbf{A}1_{v}
		} 
		\ket{x} \\
	& = &
		\frac{1}{\sqrt{2^{n}}} 
		\sum_{x}
			(-1)^{
				\frac{1}{2}
				x \cdot \mathbf{A} x
			}
		\ket{x} \\
	& = &
		\ket{G}
\end{eqnarray}

We have shown that the operators $S_v$ stabilize $\ket{G}$.  It is also easy to see that the $S_v$ operators are independent: any one cannot be obtained by multiplying together others.  They also commute with each other.  We then have $n$ independent, commuting $n$-qubit Pauli operators which stabilize a 1-dimensional space \cite{Gottesman:1997:Stabilizer-Codes-and-Quantum-Error-Correction}.

Operationally, graph states are constructed by beginning with the qubits in the state $\ket{+}^{\otimes n}$ and applying $\ctl{Z}$ gates on vertices $u,v$ whenever $(u,v) \in E$.

The above reasoning will be important later on.  In particular, we can apply~\eqref{eq:adjmatrixadd} and~\eqref{eq:adjmatrixadd2} repeatedly over all $v$ such that $y_v = 1$ to prove the following lemma.

\begin{lemma}
\label{lemma:adjmatrixadd}
\begin{equation}
		(-1)^{\frac{1}{2}
			(x  \oplus y)
			\cdot \mathbf{A}
			(x  \oplus y)
		+
			(x  \oplus y )
			\cdot \mathbf{A}y
			}
		=
		(-1)^{
			\frac{1}{2}
			x \cdot \mathbf{A} x
			}.
\end{equation}
\end{lemma}

The graph that we will mostly be concered with is a \emph{triangular lattice} graph.  Roughly speaking, a triangular lattice graph is a planar graph where every face is a triangle.  A small example is given in figure~\ref{fig:triangularLattice}

\begin{figure}
\begin{center}
\begin{tikzpicture}
\SetGraphUnit{2.7}
\GraphInit[vstyle=Simple]
\tikzset{VertexStyle/.style = {shape = circle,fill = black,minimum size = 5pt,inner sep=0pt}}
\Vertices[x=1,y=5]{line}{A,B,C,D}
\Vertices[x=2.4,y=3]{line}{E, F, G, H}
\Vertices[x=1,y=1]{line}{I, J, K, L}
\Edges(A,B,C,D)
\Edges(E,F,G,H)
\Edges(I,J,K,L)
\Edges(A,E,I)
\Edges(B,E,J,F,B)
\Edges(C,F,K,G,C)
\Edges(D,G,L,H,D)
\end{tikzpicture}
\end{center}
\label{fig:triangularLattice}
\caption{Triangular lattice graph}
\end{figure}
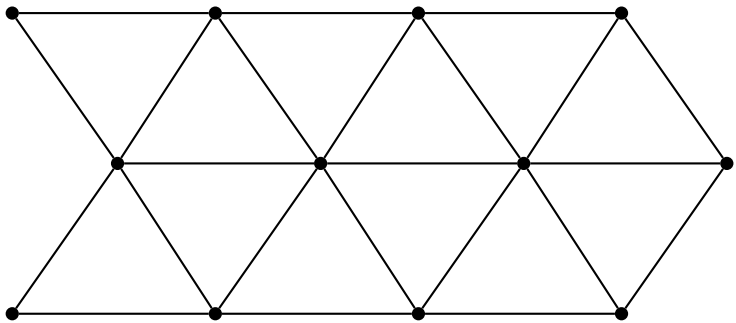

\subsection{Definition for ``closeness''}
\label{sec:defclose}
We will need to establish that the state held by the provers is ``close to'' a given graph state and that the measurements they perform are ``close to'' the ideal $X$-$Z$ plane observables.  However, there are many transformations that the provers can apply to both states and measurements which are invisible to the verifier.  In particular, the provers may add an ancilla or apply a local change of basis (simultaneously to both the state and measurements).  In fact, we will see that for the states and observables we use these are the \emph{only} undetectable transformations that they can apply\footnote
{
	The provers can also perform complex conjugation on all their states and observables, and this would also be invisible to the verifier.  However, in our case all ideal states and measurements are real, so the complex conjugation does nothing.
}.
We can account for such transformations by allowing an arbitrary isometry which undoes these transformations and presents us with the required graph state plus some arbitrary ancilla state.  We also allow for some noise by comparing states in the usual vector norm.

\begin{definition}\footnote{It is easy to see that this relation is (in the exact case) transitive and reflexive, but it is also clearly not symmetric.  Thus it is not a true equivalence relation.  However the terminology has stuck.}
We say that a multi-partite state $\ket{\psi^{\prime}}$ and observables $\{M^{\prime}\}$ are $\epsilon$-\emph{equivalent} to $\ket{\psi}$ and $\{M\}$ if there exists a local isometry $\Phi$ and a state $\ket{junk}$ such that for every $M$
\begin{equation}
	\norm{
		\Phi(M^{\prime}\ket{\psi^{\prime}}) - 
		\ket{junk}M\ket{\psi}
	}_{2} 
	\leq \epsilon.
\end{equation}
\end{definition}

Here we are thinking of ``$M$'' as both the ideal operation on $\ket{\psi}$ and as a label for the operation $M^{\prime}$.

Evidently this definition guarantees that the two systems behave like each other since isometries preserve inner products, and hence outcome probabilities.  As we shall see, it is also a necessary condition for states and measurements to behave close to the ideal graph states and $X$-$Z$ plane measurements.  Hence any other definition we could choose is at most a different characterization of the errors and in the exact case is equivalent.  The error bound used here has an operational meaning since we can quickly bound the error in outcome distributions from it.
 
There is one shortcoming of this definition, which is that it is impossible to test states or operators which contain any imaginary component in the ideal case (this restriction does not apply to the states and operators held by the provers, only to the ideal that we compare them to.)  The simple reason is that the provers may apply a complex conjugation to everything without changing the distribution of their responses to the verifier.  This transformation is not an isometry, and hence it is impossible to conclude that any system satisfies the above definition based on classical interaction alone.  It is, however, possible to extend the definition to account for this case \cite{McKague:2011:Generalized-Sel}.  We do not need to use this extended definition here since all our ideal operators and states are real.

\subsection{Modelling the provers}
An important argument in our work is that we can model the provers, even in the dishonest case, by a pure joint state held by the provers, and a collection of observables for each prover, one per possible query to that prover.  

First, it should be clear that it is not a restriction to consider pure states.  Any mixed state can be purified and the purification given to any one of the provers.  This only increases the power of the provers by giving them additional information held in the purification.

Next, since our provers will only receive one query and respond with one message, we can model their actions by a measurement.  Any pre-processing done before the measurement can be incorporated into the choice of measurement as can any post-processing.  Further, since we are not making any assumptions on the dimension of the state held by the provers, their measurements can be taken to be projective and, since the provers will always respond with $\pm 1$ the projectors can be combined into an observable without any loss of information or generality.  

Finally, we must consider how the provers will behave knowing that some of the time they will be tested and some of the time they will be asked to perform the calculation.  As well, in certain cases the provers will know for certain that they are being tested, although they will never be able to conclude that they are certainly taking part in the calculation.  The provers know in advance the list of possible query strings (there are only four) and whatever their strategy, they use some physical processes to decide on their output.  We then roll this process into the measurement observable, so  that each possible query string corresponds to a single observable which represents the entire strategy of the prover.

\section{Test for honesty}
\label{sec:selftesting}
In order to develop a test for honesty we go through several steps.  The first step is to develop a test for graph states.  This is the foundation on which we build the test for honesty.  After showing how we can verify that the provers hold onto a particular graph state we then show how to test measurements in the $X$-$Z$ plane.  Adaptive measurements built on measurements in the $X$-$Z$ plane are the next step.  Finally, we put all of the tests together into a single test and show how the probability of passing this test relates to the amount of error in an adaptive measurement performed on the same state and using the same measurements.

\subsection{Self-test for triangular cluster states}
\label{sec:graphstatetest}

In this section we develop a self-test for triangular cluster states.  The techniques used are similar to those in \cite{McKague:2010:Self-testing-gr}.  However, we make some modifications which allow for a tighter error analysis and clearer notation.  Although we give the construction for triangular cluster states only, the same techniques can be extended to work with any stabilizer state, as in \cite{McKague:2010:Self-testing-gr}.

\begin{theorem}\label{theorem:graphstatetest}
Let $G$ be a triangular lattice graph on $n$ vertices with adjacency matrix $\mathbf{A}$ and let $\epsilon > 0$. Further, suppose that for an $n$-partite state $\ket{\psi^{\prime}}$ with local measurements $X^{\prime}_{v}$ and $Z^{\prime}_{v}$ we have for each $v \in V$
\begin{equation}
	\label{eq:stabilizercondition}
	\bra{\psi^{\prime}}
		S^{\prime}_{v} 
	\ket{\psi^{\prime}} 
	\geq 
	1 - \epsilon
\end{equation}
(where $S^{\prime}_{v} = X^{\prime}_{v} Z^{\prime \mathbf{A}1_{v}}$) and for each triangle $T \subseteq V$ with characteristic vector $\tau$
\begin{equation}
	\label{eq:trianglecondition}
	-\bra{\psi^{\prime}}
		X^{\prime \tau} Z^{\prime \mathbf{A} \tau} 
	\ket{\psi^{\prime}} 
	\geq 
	1- \epsilon 
\end{equation}
then there exists a local isometry $\Phi$ and state $\ket{junk}$ such that
\begin{equation}
\label{eq:isometrycondition}
	\norm{
		\Phi\left(
			X^{\prime q} Z^{\prime p}
			\ket{\psi^{\prime}}
		\right) 
		- 
		\ket{junk} X^{q} Z^{p} \ket{G}
	} \leq 
	\left(
		2\sqrt{p \cdot p} +
		2\sqrt{2n} +
		\sqrt{|E| + n}
	\right)
	(2 \epsilon)^\frac{1}{4}
\end{equation}
for all $p,q \in \{0,1\}^{n}$.
\end{theorem}

We may interpret Theorem~\ref{theorem:graphstatetest} as follows:  for each triangular cluster state there exists a set of non-local correlations that uniquely identifies that graph state and $X$ and $Z$ measurements, up to local unitaries and additional ancillas. 

The proof can be divided into several sections.  The final goal is to construct an isometry $\Phi$ and prove that it takes the state $\ket{\psi^{\prime}}$ close to the desired graph state.  The construction for the isometry is given in terms of the $X^{\prime}$ and $Z^{\prime}$ operators on each vertex.  To bound the error we need to know how these operators behave and in particular whether they approximately anti-commute.  This is done in Lemma~\ref{lem:xzanticommute} and corollary~\ref{cor:xzanticommute}.  In the ideal case we can use the  stabilizers to show $X_{v}\ket{G} = Z^{\mathbf{A} 1_v} \ket{G}$.  In Lemma~\ref{lemma:changexz} we show that this is approximately true for the $X$'s, which will allow us to convert $X^{\prime}$s into $Z^{\prime}$s.  With these estimations in place we the proceed with the proof of Theorem~\ref{theorem:graphstatetest}.

\subsubsection{Preliminary technical estimations}
Our graph $G$ is a triangular lattice, so every vertex lies in a triangle.  For self-testing this gives a nice advantage, since it is particularly easy to show that $X^\prime$ and $Z^\prime$ anti-commute for vertices in a triangle.

\begin{lemma}\label{lem:xzanticommute}
Let $v \in V$ be a vertex in a triangle.  Under the conditions of Theorem~\ref{theorem:graphstatetest},
\begin{equation}
	\norm{
		X^{\prime}_{v} Z^{\prime}_{v} 
		\ket{\psi^{\prime}} 
		+ 
		Z^{\prime}_{v} X^{\prime}_{v}
		\ket{\psi^{\prime}}
	} \leq 
	4\sqrt{2\epsilon}.
\label{eq:xzanticommutebound}
\end{equation}

\end{lemma}
\begin{proof}
First, let $T=\{u,v,w\}$ be a triangle containing $v$.  The first part of Lemma~\ref{lemma:normbounds}, together with the conditions of Theorem~\ref{theorem:graphstatetest}, tell us 
\begin{equation}
	\norm{
		S^\prime_x \ket{\psi^\prime} 
		- 
		\ket{\psi^\prime}
	} 
	\leq 
	\sqrt{
		2 \epsilon
	}
\end{equation}
for $x \in \{u,v,w\}$, and from triangle $\tau$
\begin{equation}
	\norm{
		X^{\prime}_{u} 
		X^{\prime}_{v} 
		X^{\prime}_{w} 
		Z^{\prime \mathbf{A}1_{u}} 
		Z^{\prime \mathbf{A}1_{v}} 
		Z^{\prime \mathbf{A}1_{w}} 
		\ket{\psi^{\prime}}
		+ 
		\ket{\psi^{\prime}}
	}
	\leq
	\sqrt{
		2 \epsilon
	}.
\end{equation}
Applying the second part of Lemma~\ref{lemma:normbounds} three times to combine these, we find
\begin{equation}
	\norm{
		S^{\prime}_{u} 
		S^{\prime}_{v} 
		S^{\prime}_{w} 
		X^{\prime}_{u} 
		X^{\prime}_{v} 
		X^{\prime}_{w} 
		Z^{\prime \mathbf{A}1_{u}} 
		Z^{\prime \mathbf{A}1_{v}} 
		Z^{\prime \mathbf{A}1_{w}} 
		\ket{\psi^{\prime}}
		+ 
		\ket{\psi^{\prime}}
	} \leq 
	4\sqrt{2\epsilon}.
\end{equation}
The $Z^{\prime}$s operating on vertices outside $T$ all cancel since they appear in $S^\prime_x$ and in $Z^{\prime \mathbf{A}1_{x}}$ for some $x \in \{u,v,w\}$ and there are no $X^\prime$ operators outside the triangle. We are left with 
\begin{equation}
	\norm{
		(
			X^{\prime}_{u} 
			Z^{\prime}_{v} 
			Z^{\prime}_{w}
		)(
			Z^{\prime}_{u} 
			X^{\prime}_{v} 
			Z^{\prime}_{w}
		)(
			Z^{\prime}_{u} 
			Z^{\prime}_{v} 
			X^{\prime}_{w}
		)(
			X^{\prime}_{u}
			X^{\prime}_{v} 
			X^{\prime}_{w}
		)
		\ket{\psi^{\prime}} +
		\ket{\psi^{\prime}}
	} \leq 
	4\sqrt{2\epsilon} .
\end{equation}
By commuting operators on different subsystems past each other, we can pair up and cancel the $X^\prime_x$ and $Z^\prime_x$ for $x \in\{u,w\}$, resulting in
\begin{equation}
	\norm{
		X^{\prime}_{v} 
		Z^{\prime}_{v} 
		X^{\prime}_{v} 
		Z^{\prime}_{v} 
		\ket{\psi^{\prime}} 
		+ 
		\ket{\psi^{\prime}}
	} \leq 
	4\sqrt{2\epsilon}.
\end{equation}
Rearranging by multiplying by $Z^\prime_v X^\prime_v$, we obtain our result.
\end{proof}

Note that it is sufficient to consider a set of triangles that covers the set of vertices and hence Theorem~\ref{theorem:graphstatetest} holds for all graphs in which each vertex is contained in a triangle.  In fact, as in \cite{McKague:2010:Self-testing-gr}, it is sufficient to consider one triangle or just one edge in a connected graph, but this will give a less robust result.  Lemma~2 in \cite{McKague:2010:Self-testing-gr} shows that if $X^\prime_v$ and $Z^\prime_v$ approximately anti-commute, then so do $X^\prime_u$ and $Z^\prime_u$ for some neighbour $u$ of $v$.  Using this one can induct along paths to all vertices in a connected component. For our purposes this is unnecessary since all vertices lie in at least one triangle.

The above lemma can be generalized to products of operators, as in the following corollary.

\begin{corollary}\label{cor:xzanticommute}
Let $s,t \in \{0,1\}^{n}$.  Under the conditions of Theorem~\ref{theorem:graphstatetest}, 
\begin{equation}
	\norm{
		X^{\prime t} 
		Z^{\prime s} 
		\ket{\psi^{\prime}} 
		- 
		(-1)^{s \cdot t} 
		Z^{\prime s} 
		X^{\prime t}
		\ket{\psi^{\prime}}
	} \leq 
	4 (s \cdot t) \sqrt{2\epsilon}.
\end{equation}
where $s \cdot t$ is taken over $\mathbb{Z}$.
\end{corollary}
This can be seen by repeatedly applying Lemma~\ref{lem:xzanticommute}, once for every $v$ such that $s_v = 1 = t_v$, and using the triangle inequality.  If $s_x = 1$ but $t_x = 0$, or vice versa,  for some $x \in V$, then the single operator on vertex $x$ commutes with all other operators.

Now we consider the physical ``stabilizer generators'' $S^{\prime}_{v} = X^{\prime}_{v} Z^{\prime \mathbf{A}1_{v}}$.  The conditions of Theorem~\ref{theorem:graphstatetest} establish that they really are (close to) stabilizers of $\ket{\psi^{\prime}}$.  Next we consider products of these generators and show that they too almost stabilize $\ket{\psi^\prime}$.

\begin{lemma}\label{lemma:changexz}
Let $t \in \{0,1\}^{n}$.  Under the conditions of Theorem~\ref{theorem:graphstatetest},
\begin{equation}
	\norm{
		X^{\prime t}
		\ket{\psi^{\prime}}
		- 
		(-1)^{\frac{1}{2}t \cdot \mathbf{A}t}
		Z^{\prime \mathbf{A}t} 
		\ket{\psi^{\prime}}
	} \leq 
	\left(
		2(t \cdot \mathbf{A}t) + t \cdot t
	\right)
	\sqrt{2\epsilon}.
\end{equation}
where $t \cdot \mathbf{A}t$ and $t \cdot t$ are evaluated over $\mathbb{Z}$.

\end{lemma}
\begin{proof}
First, by Lemma~\ref{lemma:normbounds} we find
\begin{equation}
	\norm{
		\ket{\psi^{\prime}} 
		- 
		\prod_{\substack{
				v \in V \\
				t_{v} = 1
			}} 
			S^{\prime}_{v} 
			\ket{\psi^{\prime}}
	} \leq 
	(t \cdot t) \sqrt{2 \epsilon}.
	\label{eq:lemma2first}
\end{equation}
The right term in the norm can be expanded as
\begin{equation}
	\prod_{\substack{
			v \in V \\
			t_{v} = 1
		}} 
		S^{\prime}_{v} 
		\ket{\psi^{\prime}}
	= 
	\prod_{\substack{
			v \in V \\
			t_{v} = 1
		}} 
		X^{\prime v} 
		Z^{\prime \mathbf{A}1_{v}} 
		\ket{\psi^{\prime}}.
\end{equation}
We fix an ordering $<$ on $V$, and evaluate the product according to that ordering.  Thus if $t_{v} = t_{u} = 1$ and $u < v$ then $S^{\prime}_{u}$ appears in the product to the left of $S^{\prime}_{v}$.  Now suppose that $\mathbf{A}_{uv} = 1$.  Then $Z^{\prime}_{u}$ in $S^\prime_v$ appears to the right of the only occurrence of $X^{\prime}_{u}$ in $S^\prime_u$.  We may commute $Z^{\prime}_{u}$ to the right past all remaining operators on the $u$ system, so that $Z^{\prime}_{u}$ appears to the right of all $X^{\prime}$ operators.  The opposite is true if $v > u$, in which case we may commute $Z^{\prime}_{u}$ to the left, and it appears to the left of all $X^{\prime}$ operators.  Thus we may write the above as
\begin{equation}
    	\prod_{\substack{
			v \in V \\
			t_{v} = 1
		}} 
		S^{\prime}_{v} 
		\ket{\psi^{\prime}}
	= 
	\prod_{\substack{
			t_{u} = 1 \\
			u > v}} 
		Z^{\prime \mathbf{A}_{v,u}}_{v}  
		X^{\prime t} 
	\prod_{\substack{
			t_{u} = 1 \\ v > u
			}} 
		Z^{\prime \mathbf{A}_{v,u}}_{v} 
		\ket{\psi^{\prime}}.
\end{equation}
Let $\mathbf{A}^{L}$ be the lower triangular part of $\mathbf{A}$ (with 0s elsewhere) and $\mathbf{A}^{U}$ the upper triangular part.  Then we may rewrite the above as
\begin{equation}
	\prod_{\substack{
			v \in V \\
			t_{v} = 1
		}} 
		S^{\prime}_{v} 
		\ket{\psi^{\prime}}
	= 
	Z^{\prime A^{U} t} 
	X^{\prime t} 
	Z^{\prime A^{L}t} 
	\ket{\psi^{\prime}}.
\end{equation}
Using corollary~\ref{cor:xzanticommute} with $s = \mathbf{A}^{L}t$ and multiplying on the left by the unitary $Z^{\prime \mathbf{A}^{U} t}$ we obtain 
\begin{equation}
	\norm{
		Z^{\prime \mathbf{A}^{U} t} 
		X^{\prime t} 
		Z^{\prime \mathbf{A}^{L}t}  
		\ket{\psi^{\prime}}
		- 
		(-1)^{t \cdot \mathbf{A}^{L} t }
		Z^{\prime \mathbf{A}^{U} t} 
		Z^{\prime \mathbf{A}^{L}t} 
		X^{\prime t} 
		\ket{\psi^{\prime}}
	} \leq 
	4 (t \cdot \mathbf{A}^{L} t)
	\sqrt{2\epsilon}.
\end{equation}
Noting that 
$
	\mathbf{A}^{U}t + 
	\mathbf{A}^{L}t 
	= 
	\mathbf{A} t
$
and 
$
	t \cdot \mathbf{A}^{L} t 
	= 
	\frac{1}{2} 
	(t \cdot \mathbf{A} t)
$
since $A$ is symmetric, this becomes
\begin{equation}
	\norm{
	\prod_{\substack{
			v \in V \\
			t_{v} = 1
		}} 
		S^{\prime}_{v} 
		\ket{\psi^{\prime}}
		- 
		(-1)^{
			\frac{1}{2}
			(t \cdot \mathbf{A} t)
		}
		Z^{\prime \mathbf{A}t} 
		X^{\prime t} 
		\ket{\psi^{\prime}}
	} \leq 
	4 (t \cdot \mathbf{A}^{L} t)
	\sqrt{2\epsilon}.
\end{equation}
Finally we apply the triangle inequality along with~\eqref{eq:lemma2first} to find
\begin{equation}
	\norm{
		\ket{\psi^{\prime}} 
		- 
		(-1)^{
			\frac{1}{2}
			(t \cdot \mathbf{A} t)
		}
		Z^{\prime \mathbf{A}t} 
		X^{\prime t} 
		\ket{\psi^{\prime}}
	} \leq  
	\left(
		2(t \cdot \mathbf{A}t) + t \cdot t
	\right) 
	\sqrt{2\epsilon}
\end{equation}
which is transformed into the desired result by multiplying by $
    (-1)^{
        \frac{1}{2}
        (t \cdot \mathbf{A} t)
    }
    Z^{\prime \mathbf{A}t} 
$.
\end{proof}
%

\subsubsection{Proof of Theorem~\ref{theorem:graphstatetest}}

We are now in a position to prove Theorem~\ref{theorem:graphstatetest}.  This is done by giving a construction for $\Phi$ and using the above lemmas to prove that it has the necessary properties.

\begin{proof}
\begin{figure}[h]
\[
\Qcircuit @C=0.5cm @R=0.5cm{
	 & & \qw & \qw &\qw  & \qw & \qw & \qw \\
	\lstick{\ket{\phi_{+}}} \ar@{-}[ur] \ar@{-}[dr] & & & & \\
	 & & \ctrl{1} & \gate{H} & \ctrl{1}  & \gate{H} & \ctrl{1} &\qw \\
	\lstick{\ket{input}_{v}} &  & \gate{X^{\prime}_{v}} & \qw & \gate{Z^{\prime}_{v}} & \qw & \gate{X^{\prime}_{v}} & \qw\\
	}
\]
\caption{Circuit for $\Phi_{v}$}
\label{fig:epr_local_unitary_circuit}
\end{figure}
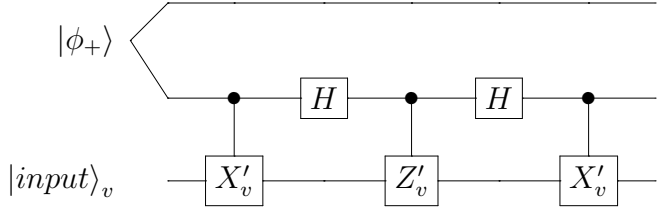

We will use $\Phi_{v}$ as defined in figure~\ref{fig:epr_local_unitary_circuit}.  The circuit is modified from that used in \cite{McKague:2010:Self-testing-gr,McKague:2012:Robust-self-tes} and earlier works, differing in the state of the ancilla.  Whereas we use an entangled pair of qubits $\ket{\phi_+}$, previous works used $\ket{0}$.  We also add an initial $\ctl{X}$ gate which was not needed when the initial state was $\ket{0}$.  When $X^{\prime}_{v}$ and $Z^{\prime}_{v}$ are the Pauli $X$ and $Z$ gates the circuit is clearly a SWAP gate. The idea is to swap the hidden qubit in the input wire with an explicit qubit.  As we shall see, the use of a maximally entangled pair of qubits in the ancilla wires allows for a tighter robustness analysis than is possible with the earlier version of the circuit.

The structure of the proof is a sequence of chained inequalities between states $\ket{\psi_{j}}$  and $\ket{\psi_{j+1}}$ for $j =1 \dots 4$ defined below.  We then use the triangle inequality to find the total distance.

We define $\Phi = \bigotimes_{v \in V} \Phi_v$, and $\ket{\psi_{1}} = \Phi\left(X^{\prime q}Z^{\prime p}\ket{\psi^{\prime}}\right)$ to be the state after the isometry $\Phi$ is applied.  Before the circuit is applied the state is 
\begin{equation}    
    X^{\prime q} 
    Z^{\prime p}
    \ket{\psi^{\prime}}
    n\ket{\phi_+}^{\otimes n} 
    = 
    \frac{1}{\sqrt{2^n}}
    \sum s 
        X^{\prime q}
        Z^{\prime p}         
        \ket{\psi^{\prime}}\ket{ss}.
\end{equation}
Applying the first $\ctl{X^\prime}$ gate yields
\begin{equation}
    \frac{1}{\sqrt{2^n}}
    \sum_s
        X^{\prime s}
        X^{\prime q} 
        Z^{\prime p} 
        \ket{\psi^{\prime}}\ket{ss}.
\end{equation}
Next we multiply by the Hadamard gates and apply Lemma~\ref{lemma:hadamardn} to obtain
\begin{equation}
    \frac{1}{\sqrt{2^{2n}}}
    \sum_{s,t} 
        (-1)^{t \cdot s}
        X^{\prime s \oplus q}
        Z^{\prime p} 
        \ket{\psi^{\prime}}\ket{st}.
\end{equation}
The $\ctl{Z^\prime}$ gate produces
\begin{equation}
    \frac{1}{\sqrt{2^{2n}}}
    \sum_{s,t} 
        (-1)^{t \cdot s}
        Z^{\prime t}
        X^{\prime s \oplus q}
        Z^{\prime p} 
        \ket{\psi^{\prime}}\ket{st}
\end{equation}
then another round of $\ctl{X^\prime}$ and Hadamard gates yields
\begin{equation}
	\ket{\psi_{1}}
	= 
	\frac{1}{\sqrt{2^{3n}}}
	\sum_{s,t,u}
		(-1)^{t \cdot(s \oplus u)}
		X^{\prime u}
		Z^{\prime t} 
		X^{\prime s \oplus q} 
		Z^{\prime p} 
		\ket{\psi^{\prime}} 
		\ket{s u}
\end{equation}
where $s,t,u \in \{0,1\}^{n}$.  The next step is to move $Z^{\prime p}$ to the left using corollary~\ref{cor:xzanticommute}, obtaining 
\begin{equation}
	\ket{\psi_{2}} 
	= 
	\frac{1}{\sqrt{2^{3n}}}
	\sum_{s,t,u}
		(-1)^{t \cdot (s \oplus u)}
		(-1)^{p \cdot (s \oplus q)}
		X^{\prime u}
		Z^{\prime t \oplus p} 
		X^{\prime s \oplus q} 
	\ket{\psi^{\prime}} 
	\ket{s u}.
\end{equation}
Next we move the combined $Z^{\prime}$s back to the right using corollary~\ref{cor:xzanticommute} again:
\begin{equation}
	\ket{\psi_{3}} 
	= 
	\frac{1}{\sqrt{2^{3n}}}
	\sum_{s,t,u}
		(-1)^{t \cdot (q \oplus u)}
		X^{\prime u \oplus s \oplus q} 
		Z^{\prime t \oplus p} 
		\ket{\psi^{\prime}} 
		\ket{s u}.
\end{equation}
Moving $Z^{\prime}$ to the left past the combined $X^\prime$s using corollary~\ref{cor:xzanticommute} one last time, we define $\ket{\psi_4}$ by
\begin{equation}
	\frac{1}{\sqrt{2^{3n}}}
	\sum_{s,t,u}
		(-1)^{t \cdot s}
		(-1)^{p \cdot (u \oplus s \oplus q)}
		Z^{\prime t \oplus p}	
		X^{\prime u \oplus s \oplus q} 
		\ket{\psi^{\prime}} 
		\ket{s u}.	
\end{equation}
Now we make two changes of variable, $t \mapsto t \oplus p$ and $u \mapsto u \oplus q$, arriving at
\begin{equation}
	\ket{\psi_{4}} 
	= 
	\frac{1}{\sqrt{2^{3n}}}
	\sum_{s,t,u}
		(-1)^{t \cdot s}
		(-1)^{p \cdot u }
		Z^{\prime t} 
		X^{\prime u \oplus s} 
		\ket{\psi^{\prime}} 
		\ket{s}
		\ket{u \oplus q}.
\end{equation}
Next we replace the $X^{\prime}$s with $Z^{\prime}$s using Lemma~\ref{lemma:changexz}
\begin{equation}
	\ket{\psi_5}
	=
	\frac{1}{\sqrt{2^{3n}}}
	\sum_{s,t,u}
		(-1)^{t \cdot s}
		(-1)^{p \cdot u }
		(-1)^{\frac{1}{2}
			(u \oplus s) \cdot \mathbf{A} (u \oplus s)}
		Z^{\prime t \oplus A(u \oplus s)} 
		\ket{\psi^{\prime}} 
		\ket{s}
		\ket{u \oplus q}.
\end{equation}
Changing variable $t \mapsto t \oplus \mathbf{A}(u \oplus s)$ we get
\begin{equation}
	\frac{1}{\sqrt{2^{3n}}}
	\sum_{s,t,u}
		(-1)^{t \cdot s}
		(-1)^{p \cdot u }
		(-1)^{s \cdot \mathbf{A}(u \oplus s)}
		(-1)^{\frac{1}{2}
			(u \oplus s)\cdot \mathbf{A} (u \oplus s)}
		Z^{\prime t} 
		\ket{\psi^{\prime}} 
		\ket{s}
		\ket{u \oplus q}
\end{equation}
and applying Lemma~\ref{lemma:adjmatrixadd} cleans this up to 
\begin{equation}
	\frac{1}{\sqrt{2^{3n}}}
	\sum_{s,t,u}
		(-1)^{t \cdot s}
		(-1)^{p \cdot u }
		(-1)^{\frac{1}{2}
			u \cdot \mathbf{A} u}
		Z^{\prime t} 
		\ket{\psi^{\prime}} 
		\ket{s}
		\ket{u \oplus q}
\end{equation}
after which the state factorizes:
\begin{eqnarray}
	\nonumber
	\ket{\psi_{5}} 
	& = & 
		\frac{1}{\sqrt{2^{3n}}}
		\sum_{s,t}
			(-1)^{t \cdot s}
			\ket{\psi^{\prime}} 
			\ket{s}
		\sum_{u}
			(-1)^{p \cdot u}
			(-1)^{\frac{1}{2} u \cdot \mathbf{A}u}
			\ket{u \oplus q}
	 \\
	& = & 
	\left(
		\frac{1}{2^{n}}
		\sum_{s,t}
			(-1)^{t \cdot s}
			Z^{\prime t}
			\ket{\psi^{\prime}} 
			\ket{s}
	\right)
	X^q Z^p 
	\ket{G}
\end{eqnarray}

In the case where $\epsilon = 0$ the above essentially gives the entire proof.  The remainder of the proof estimates the error at each step.

First we note that all the states above are normalized.  In all cases this is easy to prove.  We show the calculation for $\ket{\psi_{2}}$, with the others proceeding similarly.  We find
\begin{multline}
	\braket{\psi_{2}}{\psi_{2}} 
	= 
	\frac{1}{2^{3n}}
	\sum_{\substack{
		s, s^\prime \\
		t, t^{\prime} \\
		u, u^\prime}}
		(-1)^{
			t  \cdot (s \oplus u)
		}
		(-1)^{
			t^\prime  \cdot (s^\prime \oplus u^\prime)
		}
        (-1)^{
            p \cdot (s \oplus s^\prime)
        }
\\
		\bra{\psi^{\prime}}
		X^{\prime s \oplus q} 
		Z^{\prime t \oplus p} 
		X^{\prime u}
		X^{\prime u^\prime} 
		Z^{\prime t^{\prime} \oplus p}
		X^{\prime s^\prime \oplus q} 
		\ket{\psi^{\prime}}
		\braket{su}{s^\prime u^\prime}.
\end{multline}
The $\braket{su}{s^{\prime}u^{\prime}}$ factor implies that $u^{\prime} = u$ and $s^{\prime} = s$ for all non-zero terms so
\begin{multline}
	\braket{\psi_{2}}{\psi_{2}} 
	= 
	\frac{1}{2^{3n}}
	\sum_{s, t, t^{\prime} u}
		(-1)^{
			(t  \oplus t^\prime) 
			\cdot 
			(s \oplus u)
		}
		\bra{\psi^{\prime}}
		X^{\prime s \oplus q} 
		Z^{\prime t \oplus p} 
		X^{\prime u}
		X^{\prime u} 
		Z^{\prime t^{\prime} \oplus p}
		X^{\prime s \oplus q} 
		\ket{\psi^{\prime}}.
\end{multline}
The $X^{\prime u}$ operators square to the identity, and subsequently so do the $Z^{\prime p}$s.  
\begin{equation}
	\frac{1}{2^{3n}}
	\sum_{s, t, t^{\prime} u}
		(-1)^{
			(t  \oplus t^\prime) 
			\cdot 
			(s \oplus u)
		}
		\bra{\psi^{\prime}}
		X^{\prime s \oplus q} 
		Z^{\prime t \oplus t^{\prime}}
		X^{\prime s \oplus q} 
		\ket{\psi^{\prime}}.
\end{equation}
We then make a change of variable $t^{\prime} \mapsto t^{\prime} \oplus t$ and break $(-1)^{t \cdot (s \oplus u)}$ into $(-1)^{t \cdot s} (-1)^{t \cdot u}$.  The summand no longer depends on  $t^\prime$ so we can omit it from the summation, multiplying by $2^n$ instead.  We also bring the summation over $u$ inside, forming an inner sum.
\begin{equation}
	\frac{1}{2^{2n}}
	\sum_{s,t}
		(-1)^{t \cdot s}
		\left(
			\sum_{u}
			(-1)^{t \cdot u}
		\right)
		\bra{\psi^{\prime}}
			X^{\prime s \oplus q}
			Z^{\prime t}
			X^{\prime s \oplus q}
		\ket{\psi^{\prime}}.
\end{equation}
Lemma~\ref{lemma:sumoverstrings} says that the inner sum is 0 except when $t = 0$, so we can drop the $Z^{\prime t}$s and the summation over $t$, and $t \cdot s = 0$.   Then the $X^{\prime s \oplus q}$s then square to the identity.  We are left with $\frac{1}{2^{n}}\sum_{s} \braket{\psi^{\prime}}{\psi^{\prime}} = 1$.

Now we estimate the distances between successive states starting with $\ket{\psi_{1}}$ and $\ket{\psi_{2}}$.  From the definition of the 2-norm,
\begin{equation}
\norm{\ket{\psi_{1}} - \ket{\psi_{2}}} = \sqrt{2 - 2 \text{Re} \braket{\psi_{1}}{\psi_{2}}}.
\end{equation}
Next we determine $\braket{\psi_{1}}{\psi_{2}}$:
\begin{multline}
	\braket{\psi_{1}}{\psi_{2}}
	= 
	\frac{1}{
		\sqrt{2^{3n}}
	}
	\sum_{\substack{
			s, s^\prime \\
			t, t^\prime \\
			u, u^\prime
		}}
		(-1)^{
				t \cdot( s \oplus u)
			}
		(-1)^{
				t^\prime \cdot (s^\prime \oplus u^\prime)
			}
		(-1)^{
				p \cdot (s^\prime \oplus q)
			}
\\
		\bra{\psi^{\prime}}
			Z^{\prime p}
			X^{\prime q \oplus s}
			Z^{\prime t}
			X^{\prime u}
			X^{\prime u^\prime}
			Z^{\prime t^{\prime} \oplus p}
			X^{\prime s^\prime \oplus q}
		\ket{\psi^{\prime}}
		\braket{su}{s^\prime u^\prime}.
\end{multline}
From the $\braket{su}{s^\prime u^\prime}$ term, we see than $s = s^\prime$ and $u = u^\prime$ in all non-zero terms.  This allows us to cancel $X^{\prime u} X^{\prime u^\prime}$ and remove the $u^{\prime}$ and $s^{\prime}$ variables.  We also pull the sum over $u$ in as an inner sum
\begin{multline}
	\braket{\psi_{1}}{\psi_{2}}
	= 
	\frac{1}{2^{3n}}
	\sum_{s,t,t^{\prime}}
	\left(
		\sum_u
		(-1)^{
			(t \oplus t^\prime) \cdot u
			}
	\right)
		(-1)^{
			(t\oplus t^{\prime})
			\cdot s
		}
		(-1)^{p \cdot(s\oplus q)}
\\        
		\bra{\psi^{\prime}}
			Z^{\prime p}
			X^{\prime q \oplus s}
			Z^{\prime t \oplus t^{\prime}}
			Z^{\prime p}
			X^{\prime s \oplus q}
		\ket{\psi^{\prime}}.
\end{multline}
Next we use Lemma~\ref{lemma:sumoverstrings} to see that the inner sum is zero except when $t \oplus t^\prime = 0$.  The terms $(-1)^{(t \oplus t^\prime) \cdot s)}$ and $Z^{\prime t \oplus t^\prime}$ then become $1$ and the identity, leaving the summand independent of $t$ and $t^\prime$.  We remove them from the sum, multiplying by $2^n$ instead.  Finally, we make the change of variable $s \mapsto s \oplus q$ to get
\begin{equation}
	\braket{\psi_{1}}{\psi_{2}}
	= 
	\frac{1}{2^{n}}
	\sum_{s}
		(-1)^{p \cdot s}
		\bra{\psi^{\prime}} 
			Z^{\prime p}
			X^{\prime s}
			Z^{\prime p}
			X^{\prime s}
		\ket{\psi^{\prime}}.
\end{equation}
Next set 
$
	\epsilon_{p, s} 
	= 	
	(-1)^{p \cdot s}
	\bra{\psi^{\prime}} 
		Z^{\prime p}
		X^{\prime s}
		Z^{\prime p}
		X^{\prime s}
	\ket{\psi^{\prime}} 
	- 
	\bra{\psi^{\prime}} 
		Z^{\prime p}
		X^{\prime s}
		X^{\prime s}
		Z^{\prime p}
	\ket{\psi^{\prime}}
$, with the second term becoming just $\braket{\psi^\prime}{\psi^\prime}$, so that the above becomes
\begin{equation}
	\braket{\psi_{1}}{\psi_{2}} 
	=
	\frac{1}{2^{n}}
	\sum_{s}
		\left(
			\braket{\psi^{\prime}}{\psi^{\prime}}
			+ 
			\epsilon_{p, s}
		\right)
	=
	1 + \frac{1}{2^{n}} 
	\sum_{s} \epsilon_{p,s}.
\end{equation}
Corollary~\ref{cor:xzanticommute} and the third part of Lemma~\ref{lemma:normbounds} give $\left|\epsilon_{p, s}\right| \leq 4(p \cdot s) \sqrt{2 \epsilon}$ and the triangle inequality then gives us
\begin{equation}
	\left| 
		\braket{\psi_{1}}{\psi_{2}} - 1 
	\right| 
	\leq 
	1 + \frac{1}{2^{n}} 
	\sum_{s}
        4(p \cdot s) \sqrt{2 \epsilon}.
\end{equation}  
Lemma~\ref{lemma:averageinnerproduct} tells us how to deal with the sum over $s$, and we write
\begin{equation}
	\left| 
		\braket{\psi_{1}}{\psi_{2}} - 1 
	\right| 
	\leq 
	2(p \cdot p) 
	\sqrt{2 \epsilon}
\end{equation}
and plugging this back into the definition of the 2-norm gives
\begin{equation}
	\norm{
		\ket{\psi_{1}} - \ket{\psi_{2}}
	} \leq 
	\sqrt{4 (p \cdot p)
	\sqrt{2 \epsilon}}.
\end{equation}
Using similar estimates we find
\begin{eqnarray}
	\norm{
		\ket{\psi_{2}} 
	- 
		\ket{\psi_{3}}
	} & \leq & 
	\sqrt{2 n \sqrt{2 \epsilon}} 
\\
	\norm{
		\ket{\psi_{3}}
	 - 
	 	\ket{\psi_{4}}
	} & \leq & 
	\sqrt{2n \sqrt{2 \epsilon}}
\end{eqnarray}
where the final sum is over $s$ and $t$ rather than $s$ with the fixed string $p$.  We appeal to the second part of Lemma~\ref{lemma:averageinnerproduct} as the last step.

The remaining distance to calculate is $\norm{\ket{\psi_4} - \ket{\psi_5}}$.  Again we proceed by way of the definition of $\norm{\cdot}$ and the inner product.
\begin{multline}
	\braket{\psi_4}{\psi_5}
	=
	\frac{1}{2^{3n}}
	\sum_{\substack{
			s, s^\prime \\
			t, t^\prime \\
			u, u^\prime}
		}
		(-1)^{t \cdot s}
		(-1)^{t^\prime \cdot s^\prime}
		(-1)^{p \cdot (u \oplus u^\prime)}
		(-1)^{\frac{1}{2} 
			(u^\prime \oplus s^\prime) \cdot
			\mathbf{A}
			(u^\prime \oplus s^\prime)
			}
\\
		\bra{\psi^\prime}
			X^{\prime u \oplus s}
			Z^{\prime t}
			Z^{\prime t^\prime \oplus \mathbf{A}(u^\prime \oplus s^\prime )}
		\ket{\psi^\prime}
		\braket{s, u\oplus q}{s^\prime, u^\prime \oplus q}	
\end{multline}
For non-zero terms $s = s^\prime$ and $u = u^\prime$.  Re-indexing by $s \mapsto s \oplus u$ we find that the above is equal to
\begin{equation}
	\frac{1}{2^{3n}}
	\sum_{s,t,t^\prime}
		\left(
			\sum_u
				(-1)^{(t \oplus t^\prime) \cdot u}
		\right)
		(-1)^{(t \oplus t^\prime) \cdot s}
		(-1)^{\frac{1}{2} 
			s \cdot
			\mathbf{A}
			s
			}
		\bra{\psi^\prime}
			X^{\prime s}
			Z^{\prime t \oplus t^\prime \oplus \mathbf{A}s}
		\ket{\psi^\prime}
\end{equation}
where we have pulled all the terms dependent on $u$ into the inner sum.  Lemma~\ref{lemma:sumoverstrings} says that this inner sum is zero except where $t \oplus t^\prime = 0$ when it is $2^n$.  Substituting these in, the above becomes
\begin{equation}
	\frac{1}{2^{n}}
	\sum_{s}
		(-1)^{\frac{1}{2} 
			s \cdot
			\mathbf{A}
			s
			}
		\bra{\psi^\prime}
			X^{\prime s}
			Z^{\prime \mathbf{A}s}
		\ket{\psi^\prime}.
\end{equation}

Now let us bound the inner product:
\begin{eqnarray}
	\left|
		1 
		-
		\braket{\psi_4}{\psi_5}
	\right| 
	& = & 
	\left|
		1
		-
		\frac{1}{2^{n}}
		\sum_{s}
			(-1)^{\frac{1}{2} 
				s \cdot
				\mathbf{A}
				s
				}
			\bra{\psi^\prime}
				X^{\prime s}
				Z^{\prime \mathbf{A}s}
			\ket{\psi^\prime}
	\right| 
\\
	& \leq &
		\frac{1}{2^{n}}
		\sum_{s}
			\left|
				1 
				-
				(-1)^{\frac{1}{2} 
					s \cdot
					\mathbf{A}
					s
					}
				\bra{\psi^\prime}
					X^{\prime s}
					Z^{\prime \mathbf{A}s}
				\ket{\psi^\prime}
			\right|
\\ 		
	& \leq &
		\frac{\sqrt{2 \epsilon}}{2^n}
		\sum_s
			2 ( s \cdot \mathbf{A} s)
			+
			(s \cdot s)
\\
	& \leq &
		\frac{|E| + n}{2}\sqrt{2 \epsilon}
\end{eqnarray}
To obtain the second line above we have used the triangle inequality.  The third line comes from taking Lemma~\ref{lemma:changexz}, multiplying on the left by $\bra{X^{\prime s}}$ and applying Lemma~\ref{lemma:normbounds}.  We then use Lemma~\ref{lemma:averageinnerproduct} to obtain the last line.  Finally we find
\begin{equation}
	\norm{\ket{\psi_4} - \ket{\psi_5}} \leq \sqrt{(|E| + n)\sqrt{2 \epsilon}}
\end{equation}

Adding all the bounds using the triangle inequality we obtain
\begin{eqnarray}
	\norm{\ket{\psi_1} - \ket{\psi_5}}
	& \leq & 
	\left(
		2\sqrt{p \cdot p} +
		2\sqrt{2n} +
		\sqrt{|E| + n}
	\right)
	(2 \epsilon)^\frac{1}{4}
\end{eqnarray}
\end{proof}

\subsection{Error bounds for non-Pauli measurements}
\label{sec:xyplanemeasurements}
In order to achieve universal computation we need to have measurements other than just $X$ and $Z$.  It suffices to have $X$-$Z$ plane measurements.  Let us define
\begin{equation}
	R_{v}(\theta) 
	= 
	\cos\theta\, X_{v} + 
	\sin\theta\, Z_{v}.
\end{equation}
We use the symbol $R^{\prime}_{u}(\theta)$ to denote the $\pm 1$ eigenvalue observable that the prover uses when queried with the angle $\theta$.  We do not make any prior assumption on how $R^\prime_u(\theta)$ is related to $X^\prime_u$ or $Z^\prime_u$.  Instead we will derive said relationship via the graph-state test and further measurements.

\begin{lemma}\label{lemma:xzplanemeasurements}
Under the conditions of Theorem~\ref{theorem:graphstatetest}, if we have measurements $R^{\prime}_{v}(\theta)$ and an edge $(u,v)$ such that
\begin{equation}\label{eq:xyplanecondition}
	\bra{\psi^{\prime}} 
		R^{\prime}_{v}(\theta) 
		\left( 
			\cos\theta
			Z^{\prime \mathbf{A} 1_{v}} 
			+ 
			\sin\theta
			X^{\prime}_{u}
			Z^{\prime \mathbf{A} 1_{u} \oplus 1_{v}}
		 \right)
	\ket{\psi^{\prime}}
	\geq 
	1 - \epsilon
\end{equation}
then with $\Phi$ and $\ket{junk}$ set to those in Theorem~\ref{theorem:graphstatetest},
\begin{equation}
	\norm{
		\Phi(R^{\prime}_{v}(\theta) \ket{\psi^{\prime}}) 
	- 
		\ket{junk} R_{v}(\theta)\ket{\psi}
	} \leq
	\sqrt{2(\epsilon + 2\delta)}
\end{equation}
where $\delta$ is the bound in Theorem~\ref{theorem:graphstatetest}
\end{lemma}

\begin{proof}
From Theorem~\ref{theorem:graphstatetest} we obtain $\Phi$ and $\ket{junk}$ so that

\begin{equation}
	\norm{
		\Phi(M^{\prime}\ket{\psi^{\prime}}) 
	- 
		\ket{junk}M\ket{\psi}
	} \leq 
	\delta
\end{equation}
for $M^{\prime} \in \{ Z^{\prime \mathbf{A} 1_{v}} , X^{\prime}_{u}Z^{\prime \mathbf{A} 1_{u} \oplus 1_{v}} \}$ in particular.  From the stabilizer generators $S_u$ and $S_v$ we find $X_{u}Z^{\mathbf{A} 1_{u} \oplus 1_{v}} \ket{\psi} = Z_{v} \ket{\psi}$ and $Z^{\mathbf{A} 1_{v}}\ket{\psi} = X_{v}\ket{\psi}$, hence linearity of $\Phi$ and the triangle inequality give 
\begin{multline}
	\left|\left|
		\Phi
		\left(
			\left(
				\cos\theta\, 
				Z^{\prime \mathbf{A} 1_{v}} 
				+ 
				\sin\theta\, 
				X^{\prime}_{u}
				Z^{\prime \mathbf{A} 1_{u} \oplus 1_{v}}
			\right)
			\ket{\psi^{\prime}}
		\right)
	- 
	\right. \right.
	\\
	\left. \left.
		\ket{junk}
		\left(
			\cos\theta\, X_v
			+
			\sin \theta\, Z_v
		\right)
		\ket{G} 
	\right| \right|
	\\
	\leq
	(\cos\theta + \sin\theta) \delta
\end{multline}
Using  $
	\cos\theta\, X_v
	+
	\sin \theta\, Z_v
	= 
	R_v(\theta)
$
and $\cos \theta + \sin \theta \leq 2$ this becomes
\begin{equation}
	\norm{
		\Phi
		\left(
			\left(
				\cos\theta\, 
				Z^{\prime \mathbf{A} 1_{v}} 
				+ 
				\sin\theta\, 
				X^{\prime}_{u}
				Z^{\prime \mathbf{A} 1_{u} \oplus 1_{v}}
			\right)
			\ket{\psi^{\prime}}
		\right)
	- 
		\ket{junk}
		R_v(\theta)
		\ket{G} 
	}
	\leq
	2\delta.
\end{equation}
Now since $\norm{\bra{\psi^\prime} \Phi(R^\prime_v(\theta))}_\infty = 1$,  we have
\begin{multline}
	\left|
        \Phi\left(
            \bra{\psi^\prime}
                R^\prime_v (\theta)
        \right)
        \Phi\left(
            \left(
                \cos\theta\, 
                Z^{\prime \mathbf{A} 1_{v}} 
                + 
                \sin\theta\, 
                X^{\prime}_{u}
                Z^{\prime \mathbf{A} 1_{u} \oplus 1_{v}}
            \right)
            \ket{\psi^{\prime}}
        \right)
	\right. \\
	- 
	\left.
        \Phi\left(
			\bra{\psi^\prime}
			R^\prime_v (\theta)
        \right)
		\ket{junk}
		R_v(\theta)
		\ket{G} 
	\right|
	\leq
	2\delta.
\end{multline}
$\Phi$ preserves inner products, so this becomes
\begin{multline}
	\left|
        \bra{\psi^\prime}
            R^\prime_v (\theta)
            \left(
                \cos\theta\, 
                Z^{\prime \mathbf{A} 1_{v}} 
                + 
                \sin\theta\, 
                X^{\prime}_{u}
                Z^{\prime \mathbf{A} 1_{u} \oplus 1_{v}}
            \right)
        \ket{\psi^{\prime}}
	\right. \\
	- 
	\left.
        \Phi\left(
			\bra{\psi^\prime}
			R^\prime_v (\theta)
        \right)
		\ket{junk}
		R_v(\theta)
		\ket{G} 
	\right|
	\leq
	2\delta.
\end{multline}
Using the triangle inequality and \eqref{eq:xyplanecondition} we find
\begin{equation}
	\Phi(
		\bra{\psi^{\prime}} 
		R^{\prime}_{v}(\theta)
	) 
	\ket{junk}
	R_v(\theta)
	\ket{\psi}
	\geq 
	1 - \epsilon - 2 \delta.
\end{equation}
We now apply Lemma~\ref{lemma:normbounds} to obtain the desired bound.
\end{proof}
 
The lemma says that if we can estimate the expected value for a certain operator we can bound the error on $R^{\prime}_{u}(\theta)$.  Later in section~\ref{sec:oneshottest} we will show how we can estimate said expected value.

\subsection{Error bounds for measurement patterns}
Our bounds in the previous section show that we can bound the error when applying a measurement of the form $M_{1} \otimes \dots \otimes M_{n}$, which gives a single bit of output.  However, for graph state computation we need something much more substantial since we will need to measure the subsystems in a sequence, with each basis chosen as a function of the previous outcomes.  In fact we will prove something even stronger than this.

We will consider a stronger situation where instead of trusted classical computation and classical interaction, we have some trusted quantum computation and quantum interaction with the provers.  The provers allow the basis to be chosen quantumly and they similarly return the result coherently.  We can model this by specifying that, when queried with a quantum register, prover $j$ applies 
\begin{equation}
	V^\prime_j 
	= 
	\sum_{k=0}^{m_j} 
		\proj{k}{k} 
		\otimes 
		M_{j,k}^{\prime} 
\end{equation} 
where $M_{j,k}^\prime$ corresponds to the observable that prover $j$ uses when queried with input $k \in \{0 \dots m_j\}$.  The prover then passes the control register back to the verifier and the result of the query is stored as a $\pm 1$ phase.  We will require that the prover's actions are all of this form, although they are free to choose the $M^\prime_{j,k}$ as they like.  As well, the ideal operator $V_j$ has this form, using observables $M_{j,k}$.  

Assuming\footnote
{
	Our current self-test doesn't test whether the identity is performed correctly, since we measure the identity by just ignoring a prover.  We could easily extend the test by explicitly asking the prover to measure the identity and testing whether they return $1$, but this is unnecessary except in this imaginary case of quantum control.
}
$M_{j,0}^{\prime} = I$
we can retrieve the outcome for measurement $M^\prime_{j,k}$ by preparing the state 
$
	\frac{1}{\sqrt{2}}
	\left( 
		\ket{0} + 
		\ket{k} 
	\right)
$
and observing the relative phase change in the prover's response.  Hence this model includes the original classical behaviour as a particular case.

A general circuit for the verifier-prover interaction in this stronger model is given in figure~\ref{fig:Vjcircuit}.  The verifier first applies some unitary $U_0$ to prepare its initial state, and then performs the first query to prover 1, $V^\prime_{1}$.  The verifier then applies some unitary $U_1$ to its internal state and performs the second query to prover 2, $V^\prime_2$, and so on.  The combined operation is $U_n V^\prime_n \dots U_1 V^\prime_1 U_0$.  We require that each $V^\prime_j$ is applied at most once and for convenience we suppose that they are numbered in the order in which they are applied.  In this circuit we have always used the same the control wire, which is a q-dit with dimension equal to the maximum $m_j+1$.  This is not a limitation since we can always use the same control wire by incorporating swaps into the $U$'s if necessary.  

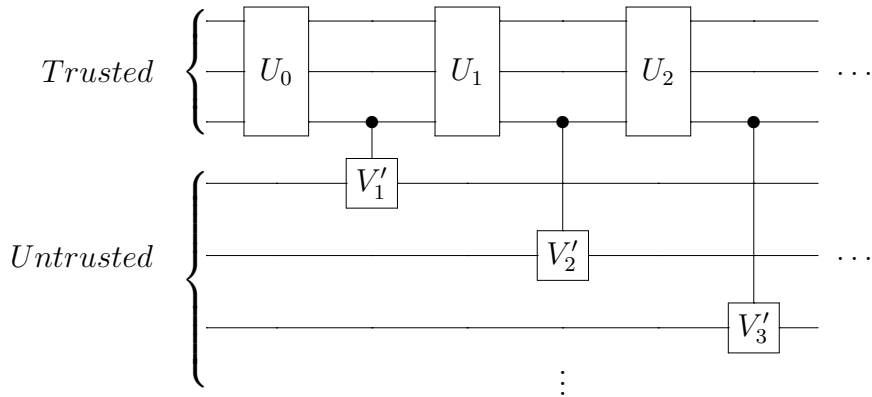
\begin{figure}[h]
\[
\Qcircuit @C=0.5cm @R=0.3cm{
&  & \multigate{2}{U_0} 	& \qw 	& \multigate{2}{U_1} & \qw & \multigate{2}{U_2} & \qw & \qw \\
	\lstick{Trusted} 	 \gategroup{1}{2}{3}{2}{.7em}{\{} & &  \ghost{U_1} & \qw 			 & \ghost{U_2} & \qw & \ghost{U_3} & \qw & \qw & \dots \\
				&	& \ghost{U_1}  			& \ctrl{1} & \ghost{U_2} & \ctrl{2} & \ghost{U_3} & \ctrl{3} & \qw \\
& &  \qw 					& \gate{V^\prime_1} & \qw & \qw & \qw & \qw & \qw \\
\lstick{Untrusted}  \gategroup{4}{2}{7}{2}{.7em}{\{} 	& &  \qw 					& \qw  & \qw	& \gate{V^\prime_2} & \qw & \qw & \qw & \dots \\
					& & \qw & \qw & \qw & \qw & \qw &  \gate{V^\prime_3} & \qw \\
					& & & & & \vdots \\ 
}
\]
\caption{Semi-trusted circuit incorporating untrusted measurements by the provers, $V^\prime_j$}
\label{fig:Vjcircuit}
\end{figure}

Let $W^\prime_0 = U_0$, and $W^\prime_j =  U_j V^\prime_{j} W^\prime_{j-1}$ for $j \geq 1$.  That is, $W^\prime_j$ represents running the circuit until the point after $U_j$ has been applied.  Similarly, let $W_j = U_{j}V_{j}W_{j-1}$ be the ideal circuit where we substitute in the ideal $V_j$ (constructed from the ideal $M_{j,k}$).

\begin{lemma}
\label{lemma:adaptivemeasurements}
Let $\ket{\psi}$, $\ket{\psi^\prime}$, $\ket{junk}$, and $\Phi = \Phi_1 \otimes \dots \otimes \Phi_n$ be given along with $M_{j,k}^\prime$ and $M_{j,k}$ ($k=0 \dots m$ and $M_{j,0}^\prime=I$) such that
\begin{equation}
\label{eq:adaptivemeasurementcondition}
\norm{\Phi\left( M_{j,k}^\prime \ket{\psi^\prime} \right)
-
\ket{junk} M_{j,k} \ket{\psi}
} \leq \delta.
\end{equation}
Further, let some $\ket{\phi}$ and $U_j$ be given where $\ket{\phi}$ contains a register of dimension at least $m+1$ and $U_j$ acts only on the $\ket{\phi}$ registers, and let, $V_j$, $V_j^\prime$, $W_j$ and $W_j^\prime$ be defined as above.  Then
\begin{equation}
	\norm{
		\Phi(W^\prime_n \ket{\psi^\prime} \ket{\phi}) 
	-
		\ket{junk} W_n\ket{\psi} \ket{\phi}
	} \leq
	(2nm+1)\delta. 	
\label{eq:adaptivemeasurementbound}
\end{equation}

\end{lemma}

The intuition is that each $V^\prime_j$ is the sum of operators $\proj{k}{k} \otimes M^\prime_{j,k}$, each of which is close to the corresponding ideal operator.  We can then use the triangle inequality to say that $V^\prime_j$ as a whole is close to its ideal counterpart.  Inducting over the depth of the circuit gives the desired result.
\begin{proof}
The proof proceeds by induction.  For the case $n=0$ we have not yet applied any untrusted gates and the conclusion is true by taking inequality~\eqref{eq:adaptivemeasurementcondition} with $k=0$ and multiplying by the trusted gate $U_0$.

Now let us suppose that~\eqref{eq:adaptivemeasurementbound} holds for $n-1$.  We start by using the bound~\eqref{eq:adaptivemeasurementcondition} with $(j,k) = (1,0)$ to get
\begin{equation}
\norm{\Phi\left(\ket{\psi^\prime} \right)
-
\ket{junk}\ket{\psi}
} \leq \delta.	
\end{equation}
For each $k \neq 0$ we multiply by the bound on both sides by $\Phi_n (M^\prime_{n,k})$ to obtain inequalities
\begin{equation}
	\norm{
		\Phi_n (M^\prime_{n,k}) \ket{junk}\ket{\psi}
	-
		\Phi_n (M^\prime_{n,k}) \Phi(\ket{\psi^\prime})
	} \leq 
	\delta.
\end{equation}
By Lemma~\ref{lemma:isometryproperites} $\Phi_n(M^\prime_{n,k})\Phi(\ket{\psi^\prime}) = \Phi(M^\prime_{n,k} \ket{\psi^\prime})$, so the state on the right above is close to $\ket{junk} M_{n,k} \ket{\psi}$ by~\eqref{eq:adaptivemeasurementcondition} with $(j,k) = (n,k)$.  Using the triangle inequality we find
\begin{equation}
	\norm{
		\Phi_n(M^\prime_{n,k}) \ket{junk}\ket{\psi}
	-
		\ket{junk} M_{n,k} \ket{\psi}
	} \leq 
	2 \delta.
\end{equation}
We introduce the register $\ket{\phi}$ and apply the ideal unitary $W_{n-1}$ to both sides in the above estimation without increasing the distance. On the left, since $\Phi_n$ and $M^\prime_{n,k}$ only operate on the $n$th subsystem, $\Phi_n(M^\prime_{n,k})$ operates only on the $n$th subsystem of $\ket{junk}\ket{\psi}$ (i.e.\ on the $n$th subsystem of $\ket{junk}$ together with the $n$th subsystem of $\ket{\psi}$).  Then since $W_{n-1}$ operates only on the trusted system and the first $n-1$ subsystems of $\ket{\psi}$, it commutes with $\Phi_n (M^\prime_{n,k})$ and $M_{n,k}$ so
\begin{equation}
	\norm{
		\Phi_n(M^\prime_{n,k}) \ket{junk} W_{n-1} \ket{\psi} \ket{\phi}
	-
		\ket{junk} M_{n,k} W_{n-1} \ket{\psi} \ket{\phi}
	} \leq
	2 \delta
\end{equation}
Now we apply the projection $\proj{k}{k}$ (used in the expression for $V^\prime_n$) to both sides, again without increasing the distance.  Hence
\begin{equation}
	\norm{
		\proj{k}{k} \otimes \Phi_n(M^\prime_{n,k}) W_{n-1} \ket{junk}\ket{\psi}
	-
		\ket{junk} \proj{k}{k} \otimes M_{n,k} W_{n-1} \ket{\psi}
	} \leq 
	2 \delta.
\end{equation}
Summing over all $k$ using triangle inequality, we apply the definitions of $V_j$ and $V^\prime_j$ to arrive at
\begin{equation}
	\norm{
		\Phi(V^\prime_n) \ket{junk} W_{n-1} \ket{\psi} \ket{\phi}
	-
		\ket{junk} V_n W_{n-1} \ket{\psi} \ket{\phi}
	} \leq
	2m \delta.
\end{equation}
Note that it is $2m \delta$ and not $2(m+1) \delta$ since the case $k=0$ has no error by assumption.  The state on the right above is almost what we want.  Now we invoke the induction hypothesis~\eqref{eq:adaptivemeasurementbound} with $n-1$ and multiply through by $\Phi(V^\prime_n)$ to get
\begin{equation}
	\norm{
		\Phi(V^\prime_n W^\prime_{n-1} \ket{\psi^\prime} \ket{\phi})
		-
		\Phi(V^\prime_n) \ket{junk} W_{n-1} \ket{\psi} \ket{\phi}
	}
	\leq 2m(n-1) \delta
\end{equation}
and applying the triangle inequality to the above two estimates we get
\begin{equation}
	\norm{
		\Phi(V^\prime_n W^\prime_{n-1} \ket{\psi^\prime} \ket{\phi})
	-
		\ket{junk} V_n W_{n-1} \ket{\psi} \ket{\phi}
	} \leq
	2mn \delta.
\end{equation}
Multiplying by the trusted gate $U_n$ (which commutes with $\Phi$) finishes the proof.
\end{proof}

In order to use this lemma we do not need the full generality of Theorem~\ref{theorem:graphstatetest}.  We only need it to apply for individual measurements rather than the full set $X^{p}Z^{q}$.

For our purposes we do not need the full strength of the lemma.  We need only know that adaptive measurements give correct outcomes. 

\begin{corollary}
\label{cor:adaptivemeasurement}
Let $\ket{\psi}$, $\ket{\psi^\prime}$, $\ket{junk}$, and $\Phi = \Phi_1 \otimes \dots \otimes \Phi_n$ be given along with $M_{j,k}^\prime$ and $M_{j,k}$ ($k=0 \dots m$ and $M_{j,0}^\prime=I$) such that
\begin{equation}
\norm{\Phi\left( M_{j,k}^\prime \ket{\psi^\prime} \right)
-
\ket{junk} M_{j,k} \ket{\psi}
} \leq \delta.
\label{eq:measurementcondition}
\end{equation}
Then for any adaptive measurement made using the $M^\prime$s, the probability of a particular outcome differs from the ideal case by at most $2(2nm+1)\delta$.
\end{corollary}

\begin{proof}
We can represent an adaptive measurement as a circuit $W_n$ as in Lemma~\ref{lemma:adaptivemeasurements}.  Hence
\begin{equation}
\label{eq:adaptivemeasurementbound2}
	\norm{
		\Phi(W^\prime_n \ket{\psi^\prime} \ket{\phi}) 
	-
		\ket{junk} W_n\ket{\psi} \ket{\phi}
	} \leq
	(2nm+1)\delta. 	
\end{equation}
To obtain the classical outcome we perform some measurement on one of the trusted subsystems.  Without loss of generality this can be a projective measurement, so let $\Pi_x$ be the projector for outcome $x$, which acts non-trivially only on the trusted subsystem.  The probability of outcome $x$ is then
\begin{equation}
	\bra{\psi^\prime} \bra{\phi} W^{\dagger\prime}_n 
	\Pi_x
	W^\prime_n \ket{\psi^\prime} \ket{\phi} .
\end{equation}

Now to estimate this probability we use~\eqref{eq:adaptivemeasurementbound2} above in two different ways.   First, multiplying on the left by $\Phi(\bra{\psi^\prime} \bra{\phi} W^{\dagger\prime}_n ) \Pi_{x}$ we get
\begin{multline}
	\left|
		\Phi(\bra{\psi^\prime} \bra{\phi} W^{\dagger\prime}_n ) 
		\Pi_x
		\Phi(W^\prime_n \ket{\psi^\prime} \ket{\phi}) 
	-
		\Phi(\bra{\psi^\prime} \bra{\phi} W^{\dagger\prime}_n ) 
		\Pi_x
		\ket{junk} W_n\ket{\psi} \ket{\phi}
	\right| \\
	\leq
	(2nm+1)\delta.
\end{multline}
Second, multiplying~\eqref{eq:adaptivemeasurementbound2} on the left by $\bra{junk} \bra{\psi} \bra{\phi} W^\dagger_n \Pi_{x}$ and then taking the adjoint of the resulting expression we obtain
\begin{equation}
	\left|
		\Phi(\bra{\psi^\prime} \bra{\phi} W^{\dagger\prime}_n ) 
		\Pi_x
		\ket{junk} W_n\ket{\psi} \ket{\phi}
	-
		\bra{\psi} \bra{\phi} W^{\dagger}_n 
		\Pi_x
		W_n \ket{\psi} \ket{\phi} 
	\right|
	\leq
	(2nm+1)\delta.
\end{equation}
Adding these together using the triangle inequality and invoking the fact that $\Phi$ preserves inner products we find
\begin{equation}
	\left|
		\bra{\psi^\prime} \bra{\phi} W^{\dagger\prime}_n 
		\Pi_x
		W^\prime_n \ket{\psi^\prime} \ket{\phi} 
	-
		\bra{\psi} \bra{\phi} W^{\dagger}_n
		\Pi_x
		W_n \ket{\psi} \ket{\phi} 	
	\right|
	\leq
	2(2mn+1)\delta.
\end{equation}
In other words, the probability of finding outcome $x$ differs from the ideal case by at most $2(2mn+1)\delta$.
\end{proof}

\subsection{A one-shot test}
\label{sec:oneshottest}
As stated, the self-testing results are not terribly useful to us.  They require knowledge of the expected value of various operators in order to draw any conclusions.  The obvious solution is to take some samples and estimate, but this would require either some independence assumptions or additional work with, for example, martingales as is done in \cite{Pironio:2010:Random-numbers-}.  Instead we will work with the contrapositive of the self-testing results:  if the state and/or some measurements are far away from the ideal, then some measurable expected value will also be far away from the ideal.  Although this is logically equivalent, instead of requiring lots of information about the various measurements, we instead are told that we just have to look for one measurement that is misbehaving.

As well, we are going to arrange our measurements in a particular way as a test for honesty.  For example, the stabilizer measurements will always return 1 for honest provers, so if we perform this measurement and we get a 1 the provers pass the test.  If result is -1 then they fail the test.  As the expected value gets close to 1, the provers will pass with probability close to 1.  If the expected value is far away from 1, the provers will fail the test with some probability.  

Now with the $R(\theta)$ measurements we do not have the same situation, but we do have something just as useful.  We can build a compound test so that the ideal honest provers pass with some probability, and no other provers can pass with a higher probability.  This is analogous to the CHSH test:  the ideal quantum strategy passes with probability $\approx 0.85$, and no other strategy achieves any higher success rate.  As well, cheating provers will pass the test with a probability that is bounded away from the quantum limit, and so we obtain a gap between the ideal and cheating strategies.  The honest provers will fail the test some of the time, but this is no problem:  we will later do some repetition so that the ideal provers will pass with an overall probability that can be made arbitrarily close to 1.

Now we give the construction for our one-shot test.  First, let $T$ be a set of triangles that covers $V$, i.e.\ each vertex in $V$ appears in at least one triangle in $T$.  The triangles will be specified by characteristic vectors $\tau$.  Let $N_{G} = 3|V| + |T|$.  Note that $N_{G} \leq 4n$ since we need no more than $n$ triangles to cover $V$.

For a graph state computation we need only two different measurement angles per vertex, $\pm \theta_{v}$.  As well, the measurement angle $\theta_{v} + \pi$ can be simulated by measuring with angles $\theta_{v}$ and flipping the outcome.  Hence there is no loss of generality by assuming that $0 \leq \theta_{v} \leq \pi$ so that $\cos\theta_{v} \geq 0$.

The test procedure is as follows:

\begin{procedure}[One-shot test for graph states and measurements]\ \\
\label{procedure:oneshottest}
\begin{enumerate}
	\item Randomly select either 
		``VERTEX'' with probability $\frac{|V|}{N_{G}}$, 
		``TRIANGLE'' with probability $\frac{|T|}{N_{G}}$, or
		``RTHETA'' with probability $\frac{2|V|}{N_{G}}$
	\item if ``VERTEX''
	\begin{enumerate}
		\item Select $v \in_{R} V$
		\item Query the provers with bases according to $S_{V} = X_v Z^{\mathbf{A} 1_v}$
		\item Accept if the product of the replies is 1, otherwise reject
	\end{enumerate} 
	\item if ``TRIANGLE''
	\begin{enumerate}
		\item Select $\tau \in_{R} T$
		\item Query the provers with bases according to $X^{\tau}Z^{A\tau}$
		\item Accept if the product of the replies is -1, otherwise reject
	\end{enumerate}
	\item if ``RTHETA''
	\begin{enumerate}
		\item Choose $t \in_{R} \{1,-1\}$ and $v \in_{R}V$ and 
			let $u$ be a vertex adjacent to $v$.  
			($u$ can be fixed ahead of time for each $v$.)
		\item Choose either 
			$X$ with probability $
				\frac{\cos\theta_{v}}
				{\cos\theta_{v} + |\sin\theta_{v}|}
			$ or
			$Z$ with probability $
				\frac{|\sin\theta_{v}|}
				{\cos\theta_{v} + |\sin\theta_{v}|}
			$
		\item if $X$
		\begin{enumerate}
			\item Query the provers with 
			$
				R_{v}(t \theta_{v})_{v}
				Z^{\prime \mathbf{A} 1_{v}} 
			$
			\item Accept if the product of the replies is 1, otherwise reject
		\end{enumerate}
		\item if $Z$
		\begin{enumerate}
			\item Query the provers with
			$
				t
				R_{v}(t \theta_{v})
				X^{\prime}_{u}
				Z^{\prime \mathbf{A} 1_{u} \oplus 1_{v}}
			$
			\item Accept if the product of the replies is 1, otherwise reject.
		\end{enumerate}
	\end{enumerate}
\end{enumerate}
\end{procedure}
To clarify, if the basis for a prover is $I$ then we simply ignore that prover, and its ``reply'' is taken to be 1.

The test is naturally grouped in to $N_{G}$ different subtests.  From the graph state test we have  $|V|$ subtests testing the ``physical stabilizers'', and  $|T|$ subtests testing the triangles.  Additionally, there are $2|V|$ ``RTHETA'' subtests, one for each choice of $v$ and $t$.  Each of these consists of two queries chosen according to some random coin.

\begin{lemma}
\label{lemma:oneshottest}
Let $n$ non-communicating quantum provers be given that each take one of four measurement bases, labelled $X$, $Z$ and $R_{v}(\pm \theta_{v})$ as inputs and measure joint state $\ket{\psi^{\prime}}$ according to operators in $\{X^{\prime}_{v}, Z^{\prime}_{v}, R^{\prime}_{v}(\theta_{v})\}$.  Then then procedure~\ref{procedure:oneshottest} accepts with probability at most
\begin{equation}
	c_{test}
	=
	\frac{
		2|V| + 
		|T| +
		\sum_{v}  
            \frac{1}{
				\cos\theta_{v} +
				|\sin\theta_{v}|
			}
        }{
		N_{G}
	}
	\quad \quad
	\text{(honest case)}
\end{equation}
and if there exist $v$ and $M \in \{X_{v}, Z_{v}, R_{v}\pm \theta_{v})_{v}\}$ such
\begin{equation}
	\label{eq:oneshottestcondition}
	\norm{
		\Phi\left(
			M^{\prime}
			\ket{\psi^{\prime}}
		\right) 
	- 
		\ket{junk} 
		M 
		\ket{G}
	} > 
	\delta
\end{equation}
then procedure~\ref{procedure:oneshottest} accepts with probability at most
\begin{equation}
	c_{test}
	-
	\frac{1}{2N_{G}}
	\left(
		\frac{
			\delta^{2}
		}{
			22 +
			25 \sqrt{n}
		}
	\right)^{4}
   	\quad \quad
	\text{(dishonest case)}
\end{equation}
\end{lemma}

\begin{proof}
{\bf Honest case.}
First let us derive the maximum probability of passing the test.  This is attained in the honest case.  The ``VERTEX'' and ``TRIANGLE'' subtests can all be passed simultaneously with probability 1 in the honest case since the observables are all in the stabilizer group of the graph state.

Let us now consider the ``RTHETA'' subtests.  First, we fix a vertex $v$.  The queries to the in the provers in this subtest can be seen as one large random variable taking values $\pm 1$ and having the expected value
\begin{multline}
 \frac{1}{2(\cos\theta_{v} + |\sin\theta_{v}|)} \\
 	\bra{\psi^{\prime}}
		R^{\prime}_{v}(\theta_{v})
		\left(
			\cos\theta_{v} 
				Z^{\prime \mathbf{A} 1_{v}} +
			\sin\theta_{v}
				X^{\prime}_{u}
				Z^{\prime \mathbf{A} 1_{u} \oplus 1_{v}}
			\right)
	\ket{\psi^{\prime}}
	+ \\
 	\bra{\psi^{\prime}}
		R^{\prime}_{v}(-\theta_{v})
		\left(
			\cos\theta_{v} 
				Z^{\prime \mathbf{A} 1_{v}} -
			\sin\theta_{v}
				X^{\prime}_{u}
				Z^{\prime \mathbf{A} 1_{u} \oplus 1_{v}}
		\right)
	\ket{\psi^{\prime}}
\end{multline}
Note the similarity to the CHSH correlation, which is obtained for $\theta_{v} = \frac{\pi}{4}$.

The honest provers (with 
$
    R^{\prime}_{v} (\theta_v)
    =
    R_{v} (\theta_v)
$)
will attain an expected value of $ \frac{1}{2(\cos\theta_{v} + |\sin\theta_{v}|)}$.  To see this, we notice that $Z^{\mathbf{A} 1_{v}} \ket{\psi} = X_v \ket{\psi}$ and $X_{u}Z^{\mathbf{A} 1_{u} \oplus 1_{v}} = Z_v \ket{\psi}$, which we obtain from the stabilizers $S_v$ and $S_u$, respectively.  Applying the definition of $R(\theta)$, the expected value becomes 
\begin{equation}
	\frac{1}{2(\cos\theta_{v} + |\sin\theta_{v}|)}
 	\bra{\psi}
	\left(
		R_v(\theta_v)^2
		+
		R_v(-\theta_v)^2
	\right)
	\ket{\psi}
	= 	\frac{1}{\cos\theta_{v} + |\sin\theta_{v}|}
\end{equation}
since $R^2_{v} (\theta_v) = I$.

Now we show that this is in fact the maximal quantum expected value.  Using a standard technique introduced by Cirel'son~\cite{Cirelson:1980:Quantum-general}, the maximum value is the same as
\begin{multline}
 	\frac{1}{2(\cos\theta_{v} + |\sin\theta_{v}|)} 
	\max_{
		\ket{\psi_{1}}, 
		\ket{\psi_{2}}, 
		\ket{\phi_{1}}, 
		\ket{\phi_{2}}
	}
\\
	\bra{\psi_{1}}
	\left(
		\cos\theta_{v} \ket{\phi_{1}} +
		\sin\theta_{v}\ket{\phi_{2}} 
	\right)+ 
\\
	\bra{\psi_{2}}
	\left(
		\cos\theta_{v} \ket{\phi_{1}} -
		\sin\theta_{v}\ket{\phi_{2}} 
	\right)
\end{multline}
where the maximization is taken over normalized states, all of dimension four.  Clearly the maximum is found when $\ket{\psi_{1}}$ is taken to be in the direction of $\cos\theta_{v}\ket{\phi_{1}} + \sin\theta_{v}\ket{\phi_{2}}$ and $\ket{\psi_{2}}$ is in the direction of $\cos\theta_{v}\ket{\phi_{1}} - \sin\theta_{v}\ket{\phi_{2}}$.  In this case the value becomes
\begin{multline}
 	\frac{1}{2(\cos\theta_{v} + |\sin\theta_{v}|)} 
	\max_{
		\ket{\phi_{1}}, 
		\ket{\phi_{2}} 
	}
\\
	\norm{
		\cos\theta_{v}
			\ket{\phi_{1}} + 
		\sin\theta_{v}
			\ket{\phi_{2}}
	} +
\\
	\norm{
		\cos\theta_{v}
			\ket{\phi_{1}} -
		\sin\theta_{v}
			\ket{\phi_{2}}
	}
\end{multline}
Expanding using the definition of $\norm{\cdot}$ we obtain

\begin{multline}
 	\frac{1}{2(\cos\theta_{v} + |\sin\theta_{v}|)} 
	\max_{
		\ket{\phi_{1}}, 
		\ket{\phi_{2}}
	}
\\
	\sqrt{
		1 + 
		2
		\cos\theta_{v}
		\sin\theta_{v}
		\text{Re}
		\braket{\phi_{1}}{\phi_{2}}
	} +
\\
	\sqrt{
		1 - 
		2
		\cos\theta_{v}
		\sin\theta_{v}
		\text{Re}
		\braket{\phi_{1}}{\phi_{2}}
	}
\end{multline}
We next use the identity $\cos\theta\sin\theta = \frac{1}{2} \sin 2\theta$ to get
\begin{multline}
 	\frac{1}{2(\cos\theta_{v} + |\sin\theta_{v}|)} 
	\max_{
		\ket{\phi_{1}}, 
		\ket{\phi_{2}}
	}
\\
	\sqrt{
		1 + 
		\sin 2 \theta_{v}
		\text{Re}
		\braket{\phi_{1}}{\phi_{2}}
	} +
\\
	\sqrt{
		1 - 
		\sin2\theta_{v}
		\text{Re}
		\braket{\phi_{1}}{\phi_{2}}
	}.
\end{multline}
Now, $\sqrt{1+a} + \sqrt{1-a} = \sqrt{(\sqrt{1+a} + \sqrt{1-a})^{2}} = \sqrt{2 + 2 \sqrt{1 - a^{2}}}$ so the above becomes
\begin{equation}
 	\frac{1}{2(\cos\theta_{v} + |\sin\theta_{v}|)} 
	\max_{
		\ket{\phi_{1}}, 
		\ket{\phi_{2}}
	}
	\sqrt{
		2 + 
		2\sqrt{
			1 -
			\left(
				\sin 2 \theta_{v} \,
				\text{Re}
				\braket{\phi_{1}}{\phi_{2}}
			\right)^{2}
		}
	}
\end{equation}
which attains the value of $\frac{1}{\cos\theta_{v} + |\sin\theta_{v}|}$ when $\braket{\phi_{1}}{\phi_{2}}= 0$.

We now have the expected value of the honest case and a matching upper bound.  The expected value of any $\pm 1$ valued random variable $X$ is related to the probability of obtaining 1 (i.e. the ``success'' probability) by
\begin{equation}
	\text{Prob}(X= 1) = 
	\frac{
		\left\langle 
			X 
		\right\rangle
	}{
		2
	} 
	+ 
	\frac{1}{2}.
\end{equation}
So the probability of success for the ``RTHETA'' portion of the test for a specific $v$ is bounded above by
\begin{equation}
	c_{test}^v := 
	\frac{1}	
	{	
		2
		(
			\cos\theta_{v} + 
			|\sin\theta_{v}|
		)
	}
	+
	\frac{1}{2}
\end{equation}

Combining this with the maximum probability of success for the ``VERTEX'' and ``TRIANGLE'' subtests, the overall maximum probability of success for any set of quantum provers, attained for honest provers, is
\begin{equation}
	c_{test} := 
	\frac{
		|V| + 
		|T| + 
		2\sum_{v}  
			c_{test}^v
	}{
		N_{G}
	} 
    = 	
    \frac{
		2|V| + 
		|T| +
		\sum_{v}  
            \frac{1}{
				\cos\theta_{v} +
				|\sin\theta_{v}|
			}
        }{
		N_{G}
	}.
\end{equation}
The factor 2 in front of the summation represents the fact that for each $v$ the ``RTHETA'' subtest occurs with probability $2/N_{G}$.

{\bf Dishonest case} 
Now that we have an upper bound, we translate the probability of success into an error bound on the expectation value of each subtest.

From now on fix a set of provers, which fixes the observables and state.  Suppose that the provers pass the test with probability $c_{test} - \frac{\epsilon}{2N_{G}}$.  Then each ``VERTEX'' or ``TRIANGLE'' subtest passes with probability at least $1-\epsilon/2$, which is obtained when all the error happens on a single subtest.  This means that the expected value for the corresponding random variable is $1-\epsilon$ and the conditions for Theorem~\ref{theorem:graphstatetest},~\eqref{eq:stabilizercondition} and~\eqref{eq:trianglecondition} are satisfied.  Hence for $M \in \{X_v, Z_v\}$ the left side of~\eqref{eq:oneshottestcondition} is bounded above by
\begin{eqnarray}
\delta_1 & := &
	\left(
		2 \sqrt{p \cdot p}
		+
		2 \sqrt{2 n}
        +
		\sqrt{|E| + n}
	\right)
	(2 \epsilon)^\frac{1}{4}
\\
	& \leq &
	2^{\frac{5}{4}}
    \left(
		1
		+
		(1 + \sqrt{2})\sqrt{n}
	\right)
	\epsilon^\frac{1}{4}
\\ & \leq &
	2.5 \left(
		1 + 
		2.5\sqrt{n}
	\right)
	\epsilon^\frac{1}{4}.
\end{eqnarray}
We have used the estimations $p \cdot p \leq 1$, since this is all we need to apply corollary~\ref{cor:adaptivemeasurement}, and $|E| \leq 3n$, since we are using a triangular cluster state which has a maximum degree of 6.  

For the ``RTHETA'' subtests, fix a $v$.  As above, the expected value for the corresponding $\pm 1$ random variable is at least $c_{test}^{v} - \epsilon$.  Hence the conditions of Lemma~\ref{lemma:xzplanemeasurements} are satisfied for each $v$ and $\pm \theta$.  Then for $M \in \{R_{v}\pm\theta_v)_v\}$ the left hand side of inequality~\eqref{eq:oneshottestcondition} is bounded above by 
\begin{eqnarray}
\delta_2 & := & 	
	\sqrt{
		2\epsilon 
		+
		10 \left(
			1 + 
			2.5\sqrt{n}
		\right)
		\epsilon^\frac{1}{4}
	}
\\	
	& \leq & 
	\sqrt{
		22 +
		25 \sqrt{n}
		}
	\epsilon^\frac{1}{8}	
\end{eqnarray}
where we have used $\epsilon \leq \epsilon^\frac{1}{4}$ for $0 \leq \epsilon \leq 1$.  When $\epsilon \leq 1$ the error for the $R^{\prime}_{v}(\pm \theta_v)$ will be larger then for $X$ or $Z$, so we will use 
\begin{equation}
\delta  = 
	\sqrt{
		22 +
		25 \sqrt{n}
		}
	\epsilon^\frac{1}{8}	
\end{equation}

We have just shown that if the provers pass the test with probability at least $c_{test} - \frac{\epsilon}{2N_{G}}$ then the left side of~\eqref{eq:oneshottestcondition} is bounded by $\delta$ as above (i.e.~\eqref{eq:oneshottestcondition} is false for all $M$).  This is the contrapositive of our desired result which is that, if~\eqref{eq:oneshottestcondition} is true for some $M$, then the probability of passing is at \emph{most} $c_{test} - \frac{\epsilon}{2N_{G}}$.  So we need only solve for $\epsilon$ in terms of $\delta$. We find
\begin{equation}
	\epsilon
	\geq 
	\left(
		\frac{
			\delta^{2}
		}{
			22 +
			25 \sqrt{n}
		}
	\right)^{4}.
\end{equation}
Now the probability of passing is at most $c_{test} - \frac{\epsilon}{2N_{G}}$, which is bounded above by
\begin{equation}
	\frac{
		2|V| + 
		|T| +
		\sum_{v}  
			\frac{1}{(
				\cos\theta_{v} +
				|\sin\theta_{v}|
			)}
	}{
		N_{G}
	} 
    -
	\frac{1}{2N_{G}}
	\left(
		\frac{
			\delta^{2}
		}{
			22 +
			25 \sqrt{n}
		}
	\right)^{4}
\end{equation}
\end{proof}

This one-shot test gives us an error bound on the states and measurements.  Combining this with Lemma~\ref{lemma:adaptivemeasurements} we can relate the probability of passing the test to the error in an adaptive measurement, i.e.\ our final measurement-based quantum computation.

\begin{corollary}
\label{cor:oneshottestadaptivemeasurements}
Let a set of quantum provers be given where prover $v$ takes inputs in $\{X_v, Z_v, R_{v}\pm (\theta_v)_v\}$ and outputs $\pm 1$.  For honest provers, procedure~\ref{procedure:oneshottest} accepts with probability 
\begin{equation}
	c_{test}
	=
	\frac{
		2|V| + 
		|T| +
		\sum_{v}  
			\frac{1}{(
				\cos\theta_{v} +
				|\sin\theta_{v}|
			)}
	}{
		N_{G}
	} 
\end{equation}
For general provers, if for any adaptive measurement pattern the probability of any outcome on the provers' final outcome differs from the ideal by more then $\delta$ then procedure~\ref{procedure:oneshottest} accepts with probability no more than
\begin{equation}
	s_{test} = 
	\frac{
		2|V| + 
		|T| +
		\sum_{v}  
			\frac{1}{2(
				\cos\theta_{v} +
				|\sin\theta_{v}|
			)}
	}{
		N_{G}
	} -
		\frac{
		\delta^{8}
		}{
		10^{17.7} n^{11}
	}
\end{equation}
\end{corollary}

\begin{proof}
We will again prove the contrapositive of the desired statement.  We would like to show that the outcome of any measurement pattern differs from the ideal by no more than $\delta$.  By corollary~\ref{cor:adaptivemeasurement}, if we achieve error less than $\delta^{\prime} = \frac{\delta}{2(8n + 1)}$ on equation~\eqref{eq:measurementcondition} then we achieve our goal, since $m=4$ here.  Lemma~\ref{lemma:oneshottest} says that we can in turn achieve this level of error if the provers pass with probability no more than
\begin{equation}
	s_{test} = c_{test} - \epsilon
\end{equation}
with
\begin{eqnarray}
	\epsilon 
	 & = &
	\frac{1}{2N_{G}}
	\left(
		\frac{
			\delta^{\prime 2}
		}{
			22 +
			25 \sqrt{n}
		}
	\right)^{4}
\\	
	& \geq &
	\frac{1}{4n}
	\left(
		\frac{
			\delta^{2}
		}{
			8
			(
				8n + 1
			)^{2}
			(
				22 +
				25 \sqrt{n}
			)
		}
	\right)^{4}
\\
	& \geq &
	\frac{1}{8n}
	\left(
		\frac{
			\delta^{2}
		}{
			4
			(
				9n
			)^{2}
			(
				47 \sqrt{n}
			)
		}
	\right)^{4}
\\ 
	& \geq &
		\frac{
			\delta^{8}
		}{
		10^{17.7}
		n^{11}
		}
\end{eqnarray}
using the pessimistic bounds $N_{G} \leq 4n$ and $1 \leq \sqrt{n} \leq n$.  Hence if the provers pass with probability less than
\begin{equation}
    s_{test} = 
	\frac{
		|V| + 
		|T| +
		2 
		\sum_{v}  
			\frac{1}{(
				\cos\theta_{v} +
				|\sin\theta_{v}|
			)}
	}{
		N_{G}
	} -
	\frac{
		\delta^{8}
	}{
		10^{17.7}
		n^{11}
	}
\end{equation}
then any adaptive measurement will differ by no more than $\delta$ from our goal.  Taking the contrapositive, if some adaptive measurement differs by more than $\delta$ then the provers will pass the test with probability no more than $s_{test}$.

\end{proof}

\section{Interactive proofs}
\label{sec:interactiveproof}
We are now in a position to construct an interactive proof for any language in $\cclass{BQP}$.  To this end, let $L$ be a language in $\cclass{BQP}$.  Then from Theorem~\ref{theorem:graphstatecomputation} and the definition of $\cclass{BQP}$ for any input $x$ there exists an adaptive measurement\footnote{The adaptive measurement is a member of a uniformly generated set.
} on a polynomially sized triangular graph state such that 
\begin{itemize}
	\item If $x \in L$ then the measurement outputs ``ACCEPT'' with probability $c_{\text{calc}} \geq \frac{2}{3}$.
	\item If $x \notin L$ then the measurement outputs ``ACCEPT'' with probability $s_{\text{calc}} \leq \frac{1}{3}$.
\end{itemize}
The adaptive measurement supplies the measurements required for each vertex via angles $\theta_v$.  It also supplies the functions required for the adaptation.

The interactive proof is given by the following procedure.

\begin{procedure}
\label{procedure:interactiveproof} \
\begin{enumerate}
	\item Randomly choose ``CALCULATE'' with probability $q$ or ``TEST'' with probability $1-q$
	\item If ``CALCULATE'':
	\begin{enumerate}
		\item Query provers according to the measurement-based computation
		\item Accept if the computations accepts
	\end{enumerate}
	\item if ``TEST''
	\begin{enumerate}
		\item Perform the test for honesty given in procedure~\ref{procedure:oneshottest}
		\item Accept if the test accepts
	\end{enumerate}
\end{enumerate}
\end{procedure}

Now we calculate the optimal value of $q$.

\begin{lemma}
\label{lemma:interactiveproofoneshot}  Let $L$ be a language and $x$ in input.  Suppose we are given an adaptive measurement on a triangular cluster state on $n$ vertices which implements a measurement-based computation which decides whether $x \in L$ with error at most $\frac{1}{3}$ (i.e. $c_{calc} \geq \frac{2}{3}$ and $s_{calc} \leq \frac{1}{3}$).  Let $0 < \delta < \frac{1}{6}$ be given and set
\begin{equation}
	q 
	= 
	\frac{
		c_{test} - s_{test}
	}{
		1 + c_{test} - 
		s_{calc} - 
		s_{test} - 
		\delta
	}
\end{equation}
then
\begin{itemize}
	\item If $x \in L$ then for honest provers procedure~\ref{procedure:interactiveproof} accepts with probability at least $c_{ip}$
	\item If $x \notin L$ then for any set of provers procedure~\ref{procedure:interactiveproof} accepts with probability at most $s_{ip}$
\end{itemize}
where
\begin{equation}
	c_{ip} - s_{ip} \geq 
	\frac{
		\delta^{8}
	}{
		10^{18.8} n^{11}
	}
\end{equation}
\end{lemma}

\begin{proof}
Let $c_{\text{test}}$ be the probability of honest provers passing the test, and let $s_{\text{test}}$ be the probability of dishonest provers passing the test, given in corollary~\ref{cor:oneshottestadaptivemeasurements}.  Here, by dishonest we mean that the probability of some outcome of an adaptive measurement made using the provers differs from the honest case by more than the given $\delta$.  

Then we have two cases:
\begin{itemize}
	\item The input is in the language:  then we only care about the honest case, in which the probability of accepting is $c_{ip} \geq qc_{\text{calc}} + (1-q) c_{\text{test}}$.
	\item The input is not in the language:  then there are two subcases:
	\begin{itemize}
		\item The provers pass the test with probability at least $s_{test}$.  Then by corollary~\ref{cor:oneshottestadaptivemeasurements} the probability of accepting on the calculation is at most $s_{calc} + \delta$ and the probability of accepting on the test is at most $c_{test}$ for an overall probability of at most $q(s_{calc} + \delta) + (1-q)c_{test}$.
		\item The provers pass the test with probability less than $s_{test}$.  Then the probability of accepting on the calculation could be as high as 1, since we gain no information from the test.  The overall probability of accepting is then less than $q + (1-q)s_{test}$.
	\end{itemize}
\end{itemize}
\noindent
The two different cases in the $x \notin L$ case give two different gaps which are, in the first case
\begin{eqnarray}
	 q c_{calc} + 
	 (1-q)c_{test} - 
	 q(s_{calc} + 
	 	\delta
	 ) 
	 - 
	 (1-q)c_{test} 
\\
	 = q(
	 	c_{calc} - 
		 s_{calc} - 
		 \delta
	)
\end{eqnarray}
and in the second case
\begin{eqnarray}
	q c_{calc} + 
	(1-q)c_{test} - 
	q 
	- 
	(1-q)s_{test} 
\\
	= q(
		c_{calc} - 
		1
	) + 
	(1-q)(
		c_{test} - 
		s_{test}
	) .
\end{eqnarray}
The overall gap is the minimum of these two.  We wish to find $q$ which gives maximizes the minimum.  The two equations are just lines in $q$ which cross each other.  At the point where they are equal we find the maximum overall gap.  The crossing point is easily found to be at
\begin{equation}
	q 
	= 
	\frac{
		c_{test} - s_{test}
	}{
		1 + c_{test} - 
		s_{test} - 
		s_{calc} - 
		\delta
	}
\end{equation}
which gives a gap of
\begin{eqnarray}
	c_{ip} - s_{ip}
	& = & 
	\frac{
		(c_{calc} -s_{calc} - \delta)(c_{test} - s_{test})
	}{
		1 + c_{test} - 
		s_{test} - 
		s_{calc} - 
		\delta
	} \\
	& \geq &
	\frac{1}{12}
		\frac{
			\delta^{8}
		}{
		10^{17.7}
		n^{11}
		}
\end{eqnarray}
On the first line, we see that the denominator can be no larger than 2, and the first factor in the numerator is at least $\frac{1}{6}$ (when $\delta = \frac{1}{6}$ and $c_{calc} - s_{calc} = \frac{1}{3}$). 
So we can lower bound the gap by $\frac{1}{12}(c_{test} - s_{test})$, which is estimated in corollary~\ref{cor:oneshottestadaptivemeasurements}.
\end{proof}

We are now in a position to prove our main claim.  This is obtained by applying a standard gap amplification procedure.

\begin{procedure}
\label{procedure:gapamplification}\ 
\begin{enumerate}
	\item Perform procedure~\ref{procedure:interactiveproof} $N$ times and let the number of times the procedure accepts be $M$.
	\item If $M > N\frac{c_{ip} - s_{ip}}{2}$ then accept.
	\item Otherwise reject.
\end{enumerate}
\end{procedure}

\begin{proof}[Proof of Theorem~\ref{theorem:bqpinteractiveproofs}]
Let $L \in \cclass{BQP}$ be given along with input $x$.  Theorem~\ref{theorem:graphstatecomputation} tells us that we can find an adaptive measurement on a triangular cluster state on $n = \text{poly}(|x|)$ vertices whose outcome tells us whether $x \in L$ with error less than $\frac{1}{3}$.  From Lemma~\ref{lemma:interactiveproofoneshot} we have an interactive proof such that if $x \in L$ then we accept with probability at least $c_{ip}$ for honest provers, and if $x \notin L$ we accept with probability at most $s_{ip}$.  

Now we amplify using procedure~\ref{procedure:gapamplification}.  Provided that we use fresh randomness on each run of the interactive proof, the $N$ trials are all independent, although not necessarily identically distributed.

Let us first consider the case $x \in L$.  Then we are interested in the case of honest provers, in which case we have $N$ independent and identical Bernoulli trials with some probability of accepting $p \geq c_{ip}$.  Using Hoeffding's inequality, the probability that we mistakenly reject is bounded by 
\begin{eqnarray}
	P
	\left(
		M 
		\leq
		N
		\frac{
			c_{ip} - s_{ip}
		}{2}
	\right)
	& \leq &
	\exp 
	\left(
	-2
	 \frac{
		\left(
			Np - 
			N
			\frac{
				c_{ip} - s_{ip}
			}{2}
		\right)^2
	}{
		N
	}
	\right)
\\	
	& \leq &
	\exp
	\left(
		-\frac{
			N
			(
				c_{ip} 
				-
				s_{ip}
			)^2	
		}{2}
	\right)
\end{eqnarray}
Setting this equal to $\frac{1}{3}$ we solve for $N$ to find the minimum number of trials to achieve our desired error rate.
\begin{equation}
	N 
	\geq
	\frac{
		2 \ln 3 
	}{
		(
			c_{ip} 
			-
			s_{ip}
		)^2		
	}.
\end{equation}
Subbing in for $c_{ip} - s_{ip}$ as estimated in Lemma~\ref{lemma:interactiveproofoneshot} we obtain
\begin{equation}
	N \geq
	\frac{
		10^{37.9} n^{22}
	}{
		\delta^{16}
	}.
\end{equation}

The same number of repetitions also suffices to bound the probability of accepting when $x \notin L$ to below $\frac{1}{3}$.  The analysis is similar, however if the provers are not honest they may vary their behaviour on each trial, so the trials are not necessarily identically distributed.  However, since the probability $p$ of accepting satisfies $p \leq s_{ip}$ for every trial we can still use Hoeffding's inequality.

Note that there is ambiguity in procedure~\ref{procedure:gapamplification}.  In particular, it does not mention whether the repetition of procedure~\ref{procedure:interactiveproof} should be done serially or in parallel.  In fact this does not matter.  We can add $N$ sets of $n$ provers and query each prover once, or use one set of $n$ provers and query each one $N$ times.  Either way the fact that we use fresh randomness means that the upper bound on the success probability applies for each trial, and any correlations between trials cannot increase this bound.  If $x \notin L$ then there are no provers, whether entangled with other provers in other trials, retaining memory of past trials, or otherwise, that will force the verifier to accept with probability more than $s_{ip}$.  Hence we can specify that the repetition is accomplished in parallel using $nN$ provers, each of which, in the honest case, performs a single measurement.
 
\end{proof}

\section{Discussion}
\subsection{Assumptions}
For our result to hold, we must assume that the provers behave according to quantum mechanics.  This is because we model the provers using the quantum formalism.  Currently this appears to be a reasonable assumption since quantum mechanics has been a very successful theory.  However, we run into a problem if we wish to use this result in certain circumstances. For example, we may wish to verify that quantum mechanics generates accurate predictions for very complex systems.  In this case it becomes infeasible to classically compute predictions from the quantum model.  Quantumly, this is still possible, and we might be tempted to use an interactive proof in order to verify that our quantum computer has done the computation correctly.  However, if we use the arguments in this paper we run into a problem of circularity since we must assume that quantum mechanics is correct in order for the argument to go through, but quantum mechanics is exactly what we want to verify!  Hence it remains an important open question whether it is possible to achieve interactive proofs for problems in BQP where the prover's actions are easy for quantum computers but for which we do not assume \emph{a~priori} that quantum mechanics holds.

\subsection{Time-space trade-offs}
As discussed in the proof of Theorem~\ref{theorem:bqpinteractiveproofs}, there are space-time trade-offs.  In particular we may perform the gap amplification by repeating in parallel or serially, or some mixture of the two.  Hence we could perform $N$ repetitions on $n$ provers, each of which then performs $N$ measurements, or we can repeat in parallel with $nN$ provers, each of which performs a single measurement.

Another factor, which we have not mentioned, is the time required to build the necessary graph states.  Graph states are built by applying $\ctl{Z}$ gates, one for each edge.  At worst this can take no more than $O(n^2)$ operations.  For triangular cluster states, such as we use here, the degree of each vertex in bounded above by 6 so at most $3n$ 2-qubit operations are required (plus $n$ single-qubit state preparations).  These can be parallelized in a constant depth circuit by exploiting the localized structure of triangular cluster states.  Regardless, the state preparation can be accomplished in a polynomial number of steps.

\subsection{Simplicity}
Our construction has a somewhat remarkable property.  Since triangular cluster states are universal, the state preparation depends only on the size of the calculation.  By choosing a discrete set of operations ($X$, $Z$ and $X\pm Z$ measurements are sufficient for universality) the quantum provers are also constant, requiring only the ability to measure in some fixed set of bases.  The self-testing portion of the interactive proof then also only depends on the size of the calculation, since it depends only on the state and the measurements.  Amazingly the classical verifier is also rather simple.  It can be given a circuit as its input from which it reads off what gates to perform and simply looks up what angle to measure for that gate.  The remaining calculations are simply XORs.  This is a clear example of where the simplicity of the measurement-based quantum computing model allows for a simple analysis.

Another interesting property is that no single quantum prover has enough power to convince even \emph{itself} whether the input is in the language.  Indeed, all the quantum parts together, including the state preparation, still cannot perform even simple calculations since they can only prepare and measure some fixed state.  They lack the capacity to perform the XORs required to perform a full measurement-based calculation.  It is only when we combine the verifier, provers, and state preparation together that we obtain enough power to perform any substantial calculations.

\subsection{$X$-$Y$ plane measurements}
It should be possible, although a bit more involved, to use the usual graph state computation model \cite{Raussendorf:2001:A-One-Way-Quant,Raussendorf:2003:Measurementbasedquantum} involving $X$-$Y$ plane measurements.  There is a complication since it is possible to simulate a complex measurements $M$  by using a $M^{*}$ measurements instead (i.e. complex conjugate everything).  In other words, the provers could complex conjugate all their states and operators, and this would preserve the expected values of all operators.  However, the complex conjugation cannot be ``undone'' through an isometry.  Hence the complex conjugated provers would not satisfy equation~\eqref{eq:isometrycondition}.  

As far as correctness of the calculation is concerned, complex conjugation poses no problem since the outcome of an adaptive measurement will be identical to the desired case.  A workaround is possible through a suitable relaxation of~\eqref{eq:isometrycondition}, and details are given for the exact case in \cite{McKague:2011:Generalized-Sel}.  The remaining points are to make the techniques of \cite{McKague:2011:Generalized-Sel} robust, and generalize to $X$-$Y$ plane measurements, which can be done analogously to how $X$-$Z$ plane measurements are handled here. 

\subsection{Future work}

Our error bounds are clearly suboptimal.  In many places we have made only loose estimates which suffice for our purposes of establishing a polynomial bound, but could be made more robust.  Hence one avenue of future improvement is to tighten these bounds.

Currently our construction uses many simple provers, providing a nice complement to Reichardt et al.'s result using a constant number provers.  Much of our result could easily be adapted to the case of two provers.  The most difficult part is the graph-state test.  Likely it is not possible to prove a self-testing theorem for two provers if there are any odd cycles in the graph since it would be necessary at some point to test the entanglement across an edge with both vertices held by a single prover.  However, bipartite graph states could yet be self-tested with two provers.

{\bf Acknowledgements} Thanks to Serge Massar, Stefano Pironio, Iordanis Kerenidis, Fr\'{e}d\'{e}ric Magniez, David Hutchinson and Michael Albert for helpful discussions.  This work is funded by the Centre for Quantum Technologies, which is funded by the Singapore Ministry of Education and the Singapore National Research Foundation, and by the University of Otago.

\bibliographystyle{halphamm}
\bibliography{Global_Bibliography}

\appendix

\section{Technical lemmas for estimation}\label{appendix:technicallemmas}
We will need to make use of several easy technical results in our proofs.  We collect them here for convenience.

\begin{lemma}\label{lemma:normbounds}
Let $\ket{\psi}$ and $\ket{\phi}$ be normalized states.  Suppose $\bra{\psi} M\ket{\psi} \geq 1 - \alpha$ and $\bra{\psi} N \ket{\psi} \geq 1 - \beta$ where $M^{2} = N^{2} = I$.  Then
\begin{eqnarray}
	\norm{
		\ket{\psi} - M \ket{\psi}
	} & \leq & 
	\sqrt{2\alpha} 
\\
	\norm{
		\ket{\psi} - MN \ket{\psi}
	} & \leq & 
	\sqrt{2} 
	\left(
		\sqrt{\alpha} + 
		\sqrt{\beta}
	\right) 
\end{eqnarray}
Further, if $M$ is unitary and $\ket{\psi_{1}}$ and $\ket{\psi_{2}}$ are normalized states, then
\begin{equation}
\left| \bra{\phi} M \ket{\psi_{1}} - \bra{\phi} M \ket{\psi_{2}} \right| \leq \norm{\ket{\psi_{1}} - \ket{\psi_{2}}}.
\end{equation}
\end{lemma}

The first inequality is a straightforward applications of the definition of $\norm{\cdot}$.  The second inequality is an application of the first, along with the triangle inequality.  The last inequality is an application of the inequality $\norm{O \ket{\psi}} \leq \norm{O}_\infty \norm{\ket{\psi}}_2$ where we use the operator $O = \bra{\phi} M$.

\begin{lemma}\label{lemma:sumoverstrings}
Let $t$ be an $n$-bit string.  Then
\begin{equation}
\sum_{s \in \{0,1\}^{n}} (-1)^{s \cdot t} = 2^{n} \delta_{t}.
\end{equation}
\end{lemma}
If $t = 0$ then the summand is always 1.  If $t \neq 0$ then half the strings $s$ have inner product $0$ with $t$ and the other half have inner product $1$, so we get a sum with half the summands 1 and the other half -1.

\begin{lemma}\label{lemma:averageinnerproduct}
Let $u \in \{0,1\}^{n}$ be given and let $\mathbf{A}$ be the adjacency matrix for a graph $G = (V,E)$.
\begin{eqnarray}
	\frac{1}{2^{n}}
	\sum_{s \in \{0,1\}^{n}} 
		s \cdot u
	 & = & 
	 \frac{u \cdot u}{2} 
 \\
	\frac{1}{2^{2n}}
	\sum_{s,t \in \{0,1\}^{n}} 
		s \cdot t 
	& = & 
	\frac{n}{4} 
\\
	\frac{1}{2^n}
	\sum_{t \in \{0,1\}^n}
		t \cdot \mathbf{A} t
	& = &
	\frac{|E|}{4}
\end{eqnarray}

\end{lemma}
For the first one, the average inner product of a vector with $u$ is half the number of 1's in $u$.  The second computes this for an average $u$, which has $n/2$ 1's.  For the last one, $t \cdot \mathbf{A}t$ counts the number of edges in the induced subgraph on $S_t = \{v \in V | t_v = 1\}$.  Consider an edge $(u,v)$.  Then $(u,v)$ appears in the induced subgraph on $S_t$ whenever both ends are in $S_t$, i.e. when $t_u = t_v = 1$.  This happens for a quarter of all bit strings $t$.  Hence each edge is counted $2^{n-2}$ times for a total of $2^{n-2}|E|$.

\begin{lemma}
\label{lemma:hadamardn}
Let $x \in \{0,1\}^{n}$.  Then
\begin{equation}
	H^{\otimes n} 
	\ket{x} 
	= 
	\frac{1}{\sqrt{2^{n}}} 
	\sum_{y \in \{0,1\}^{n}} 
		(-1)^{x \cdot y} 
		\ket{y}
\end{equation}

\end{lemma}
This is a standard result in quantum computing, and can be shown using induction on $n$.

\section{Local complementation and self-testing stabilizer states}
\textbf{Local Complementation.}
Let $G$ be a graph.  We may form a new graph $G^{\prime}$ by \emph{local complementation} on a vertex $v \in V$.  This operation complements all the edges in the neighbourhood of $v$, meaning that if $a,b \in V$ are neighbours of $v$ and $(a,b) \in E$ then $(a,b)$ is removed in $G^{\prime}$ and, conversely, if $(a,b) \notin E$ then $(a,b)$ is added in $G^{\prime}$.

Local complementation is relevant to graph states because $\ket{G}$ and $\ket{G^{\prime}}$ will be related by local Clifford operations, which simply relabel Pauli measurements and outcomes.  The standard stabilizer generators for $\ket{G^{\prime}}$ are found from those of $\ket{G}$ by replacing $S_{u}$ with $S_{u}S_{v}$ for each neighbour $u$ of $v$, and exchanging $Z_{v}$ with $Y_{v}$, and $X_{u}$ with $Y_{u}$.  

\textbf{Self-testing stabilizer states.}
First, note that an arbitrary stabilizer state is equivalent to a graph state under local Clifford operations \cite{Schlingemann:2001:Stabilizer-code}.  The local Clifford operations just relabel the Pauli operators (up to $\pm 1$), so the problem of testing stabilizer states reduces to that of testing graph states.

To self-test an arbitrary graph we first divide it into connected components.  The graph state will a product state where the overall graph state is the product of the graph states on the connected components.  We may thus test connected components individually.  If a component has only one vertex then the test is trivial.  If it has two vertices then the graph state on that component is an EPR pair (up to local unitaries), which can be tested using the Mayers-Yao test \cite{Mayers:2004:Self-testing-qu} or the CHSH test \cite{Clauser:1969:Proposed-Experi,McKague:2012:Robust-self-tes}.

Now we consider a connected component with three or more vertices.  We need to show that equation~\eqref{eq:xzanticommutebound} in Lemma~\ref{lem:xzanticommute} holds for each vertex, i.e.\ the $X^\prime_v$ and $Z^\prime_v$ operators approximately anti-commute on the state.  It is shown in \cite{McKague:2010:Self-testing-gr} (Lemma 2) that if equation~\eqref{eq:xzanticommutebound} holds on vertex $v$ then it also holds (with a larger bound) on any neighbour $u$ of $v$.  Hence as long as~\eqref{eq:xzanticommutebound} holds for at least one vertex in a connected component, then it also holds for all other vertices in the component by inducting along paths.

Let us return to our component with three or more vertices.  If it contains a triangle then we test using Lemma~\ref{lem:xzanticommute}.  If not, then let $u,v \in V$ be two vertices in the component that are not adjacent.  Since they are in the same component, there is a shortest path between them of length at least 2.  Look at the first three vertices in this path.  The first and third are not adjacent, otherwise the path would be shorter.  Hence we have an induced path of length 2.  We then locally complement on the middle vertex, obtaining a triangle.  The corresponding graph state is equivalent to the original graph state under local Clifford operations, which can be absorbed into the definition of the local isometry $\Phi$.  Now we have a graph state with a triangle and we can apply Lemma~\ref{lem:xzanticommute} and \cite{McKague:2010:Self-testing-gr} (Lemma 2) to obtain the anti-commuting relation for all vertices.

Having tested all components and obtained local isometries $\Phi_j$ for each component $j$, we simply form $\Phi = \Phi_1 \otimes \dots \otimes \Phi_k$ and apply triangle equalities to obtain the final isometry and bound.
\end{document}